\documentclass[11pt]{amsart}

\usepackage[hmargin=0.8in,height=8.6in]{geometry}
\usepackage{amssymb,amsthm, times}
\usepackage{delarray,verbatim}
\usepackage{mathrsfs}
\usepackage{multirow}
\usepackage{longtable}
\usepackage{ifpdf}
\ifpdf
\usepackage[pdftex]{graphicx}
 \else
\usepackage[dvips]{graphicx}
 \fi

\usepackage{bm}

\usepackage{ifpdf}
\usepackage{color}
\definecolor{webgreen}{rgb}{0,.5,0}
\definecolor{webbrown}{rgb}{.8,0,0}
\definecolor{emphcolor}{rgb}{0.5,0.95,0.95}

\usepackage{hyperref}
\hypersetup{%
	colorlinks=true,
	linkcolor=webbrown,
	filecolor=webbrown,
	citecolor=webgreen,
	breaklinks=true}
\ifpdf \hypersetup{pdftex,
	pdfstartview=FitH, 
	bookmarksopen=true,
	bookmarksnumbered=true
} \else \hypersetup{dvips} \fi

\newcommand {\ud}{{\rm d}}

\numberwithin{equation}{section}

\newtheorem{theorem}{Theorem}[section]
\newtheorem{proposition}{Proposition}[section]

\newtheorem{remark}{Remark}[section]
\newtheorem{lemma}{Lemma}[section]

\numberwithin{remark}{section} \numberwithin{proposition}{section}
\numberwithin{corollary}{section}
\newcommand {\R}{\mathbb{R}}

\newcommand {\N}{\mathbb{N}}

\newcommand {\E}{\mathbb{E}}

\definecolor{gray}{gray}{0.6}

\begin{document}
	\title[Two-type branching processes with immigration, and the structured coalescents]{Two-type branching processes with immigration, and the structured coalescents}

	\author[M. E. Caballero]{Mar\'ia Emilia Caballero$^\dagger$}
	\thanks{$^\dagger$Instituto de Matem\'aticas, Universidad Nacional Aut\'onoma de M\'exico, \'Area de la Investigaci\'on Cient\'ifica, Circuito Exterior, Ciudad Universitaria, Coyoac\'an 04510,
		Cuidad de M\'exico, M\'exico. Email: marie@matem.unam.mx}
	
	\author[A. Gonzalez]{Adri\'an Gonz\'alez Casanova$^{*}$}
	\thanks{$^{*}$Instituto de Matem\'aticas, Universidad Nacional Aut\'onoma de M\'exico, \'Area de la Investigaci\'on Cient\'ifica, Circuito Exterior, Ciudad Universitaria, Coyoac\'an 04510,
		Cuidad de M\'exico, M\'exico. Email: adriangcs@matem.unam.mx}

	\author[J.L. P\'erez]{Jos\'e Luis P\'erez$^{**}$}\thanks{$^{**}$Departamento de Probabilidad y Estad\'istica, Centro de Investigaci\'on en Matem\'aticas, A.C. Calle Jalisco S/N
		C.P. 36240, Guanajuato, Mexico.
		Email: jluis.garmendia@cimat.mx}
	
	\maketitle
	
	\vspace{-.2in}
	\maketitle
	\begin{abstract}
		We consider a population constituted by two types of individuals; each of them can produce offspring in two different islands (as a particular case the islands can be interpreted as active or dormant individuals). We model the evolution of the population of each type using a two-type Feller diffusion with immigration, and we study the frequency of one of the types, in each island, when the size of the total population in each island is forced to be constant at a dense set of times. This leads to the solution of an SDE which we call \textit{the asymmetric two-island frequency process}. We derive properties of this process and obtain a large population limit when the total size of each island tends to infinity. Additionally, we compute the fluctuations of the process around its deterministic limit. We establish conditions under which the asymmetric two-island frequency process has a moment dual. The dual is a continuous-time two-dimensional Markov chain that can be interpreted in terms of mutation, branching, pairwise branching, coalescence, and a novel mixed selection-migration term. 
	\end{abstract} 
	
	{\small {{\it AMS 2020 subject classifications}:  60J90, 60J80, 92D15, 92D25}
		
		}
	
	\section{Introduction}
	
	The evolution of the number of individuals in a population and its genetic profile have been extensively studied. Two of the main families of processes used in this task are branching processes and frequency Wright--Fisher type processes,  known also as two-type Fleming--Viot processes. The relation between the two has constituted a formidable source of mathematical questions. Some remarkable results on this topic were made by \cite{7authors} and \cite{EM}. In \cite{EM}, Etheridge and March relate these two families by conditioning a Dawson--Watanabe superprocess to have constant mass, obtaining a  Fleming--Viot superprocess. 
	This technique peaked in \cite{7authors}, where this relation was extended by showing that self-similar branching processes are in correspondence with $\beta$-coalescents. 
	
	In the present paper, we study the relation between two-type Feller diffusions with continuous immigration and the generalized structured Wright--Fisher type models known as the \textit{two-island model with mutation and selection} \cite{NG}.

	
	Our approach to relate these prominent families of processes is different from the one in \cite{7authors} and is inspired by two seminal papers \cite{Gill73,Gill74}, where Gillespie introduced a method to model the competition between two types of individuals. His method relies strongly on his notion of \textit{culling}. He starts with the assumption that each of the  populations evolves according to  independent Feller diffusions with possibly different distributions. At particular times, which we call sampling times, the total mass is returned to a fixed level in such a way that the relative frequency of each type is preserved. It is important that the two competing populations are not assumed to have the same reproduction mechanism. This recurrent sampling was inspired by seasonal species and describes accurately the recurrent sampling that is characteristic of experimental evolution. The most prominent example of experimental evolution is the Lenski experiment  \cite{Lenski} (see  \cite{GKWY} for a probabilistic approach to recurrent sampling experiments). 
	
	
	The Gillespie recurrent sampling method (GRSM) was mathematically formalised in \cite{CGP} and extended in order to include jumps and immigration in the one-dimensional continuous-state branching process (CSBP) that models the population size of each of the types.  The GRSM is weaker than the time-change method in  \cite{7authors} since it relates branching and frequency processes only in distribution, but it is more general: in particular, it can be used on pairs of branching processes with different distributions, leading to asymmetric frequency processes.
	
	In this paper we extend the GRSM to the two-dimensional setting by applying this procedure to two-type Feller diffusions with immigration. The resulting family, which we call \textit{asymmetric two-island frequency processes} includes, as special cases, the seed bank model \cite{BGKW15,BGKW16} and the two-island model \cite{NG}. The main strength of the two-dimensional GRSM is that it allows comparing populations with different reproduction mechanisms, where each individual can be in one of two different states (for example active/dormant or island 1/island 2). Many of the members of this family have a notion of genealogy in the sense of a moment dual. We will provide conditions for moment duality and describe the moment dual processes. This manuscript generalizes the work in \cite{CGP} to the two dimensional setting.

	If we consider two-type branching processes consisting of active and dormant individuals, where active individuals can produce active and dormant individuals while dormant individuals are only able to activate, and apply the GRSM, we obtain the seed bank frequency processes with mutation and selection  \cite{BGKW15,BGKW16}. This relation between seed bank models and two-type branching process is very natural. However, this seems to be the first time that it has been formally proved in the literature. 

	\subsection{Two-type Feller diffusions with immigration}\label{preliminiaries}
We start with a population consisting of two types of individuals, namely $\textbf{x}$ and $\textbf{y}$, distributed between two islands which we refer to as $1$ and $2$. The islands can be interpreted as two different states, for example, dormant and active, or juvenile and mature individuals.
We will assume that individuals can only give birth to offspring of the same type, but the offspirng can reside in any of the two different islands. Additionally, we will consider external immigration events that contribute to the population of individuals of each type and island.

To this end, we consider that the population of individuals of type $\textbf{x}$ is described by a multitype Feller diffusion 
with immigration $\mathbf{X}:=\{\mathbf{X}_t=(X^1_t,X^2_t):t\geq 0\}$, given as the unique strong solution to the following stochastic differential equation (see for instance \cite{BLP}):
\begin{align}\label{sde_X}
	\mathbf{X}_t&=\mathbf{X}_0+\int_0^t(\beta^X+B^X\mathbf{X}_s)ds+\sum_{i=1}^2\int_0^t\sqrt{c_i^XX_s^i}dW_s^{X,i}e_i, \qquad t\geq 0,
\end{align}
where $\mathbf{X}_0:=(x^1,x^2)\in\R_+^2$, $\mathbf{W}^X:=(W^{X,1},W^{X,2})$ is a two-dimensional Brownian motion, $e_1$ and $e_2$ denote the canonical vectors in $\R^2$, and
\begin{itemize}
	\item $c^X=(c_i^X)_{i\in\{1,2\}}\in\R_+^2$.
	\item $\beta^X=(\beta_i^X)_{i\in\{1,2\}}\in\R_+^2$.
	\item $B^X=(B^X_{ij})_{i,j\in\{1,2\}}\in\R^{2\times2}_{(+)}$,
\end{itemize}
where $\R^{2\times2}_{(+)}$ denotes the space of essentially non-negative $2\times2$ matrices i.e. the set of matrices with non-negative off-diagonal entries, and $\R_+^2=[0,\infty)\times[0,\infty)$.

Heuristically, the component $X^i$ describes the total mass of the population of individuals of type $\textbf{x}$ at the island $i$ for each $i=1,2$. 
On the other hand, for each $i=1,2$, the parameter $B^X_{ii}$ is related to the rate at which "infinitesimal individuals" of type $\textbf{x}$ located at the island $i$ produce offspring of the same type at island $i$; while the variance associated to the these branching events is denoted by $c_i^X$. Additionally for $i,j=1,2$ with $i\not=j$, the parameter $B^X_{ij}$ describes the rate at which "infinitesimal individuals" of type $\textbf{x}$ residing at the island $j$ migrate to island $i$. Finally, the parameter $\beta^X_i$ depicts the rate at which immigration events occur that contribute to the total mass of the population of individuals of type $\textbf{x}$ at the island $i$ for each $i=1,2$.

Similarly, the total mass of individuals of type $\textbf{y}$ is given by a multitype Feller diffusion with immigration $\mathbf{Y}:=\{\mathbf{Y}_t=(Y^1_t,Y^2_t):t\geq 0\}$, defined as the unique strong solution to the following stochastic differential equation:
\begin{align}\label{sde_Y}
	\mathbf{Y}_t&=\mathbf{Y}_0+\int_0^t(\beta^Y+B^Y\mathbf{Y}_s)ds+\sum_{i=1}^2\int_0^t\sqrt{c_i^YY_s^i}dW_s^{Y,i}e_i, \qquad t\geq 0.
\end{align}
where $\mathbf{Y}_0:=(y^1,y^2)\in\R_+^2$ and $\mathbf{W}^Y:=(W^{Y,1},W^{Y,2})$ is a two-dimensional Brownian motion independent from $\mathbf{W}^X$. Additionally,
\begin{itemize}
	\item $c^Y=(c_i^Y)_{i\in\{1,2\}}\in\R_+^2$.
	\item $\beta^Y=(\beta_i^Y)_{i\in\{1,2\}}\in\R_+^2$.
	\item $B^Y=(B^Y_{ij})_{i,j\in\{1,2\}}\in\R^{2\times2}_{(+)}$.
\end{itemize}
As in the case of individuals of type $\textbf{x}$, the component $Y^i$ gives the total mass of individuals of type $\textbf{y}$ at the island $i$ for $i=1,2$. The description of the parameters appearing in \eqref{sde_Y} is analogous to those of the process $\mathbf{X}$.

We remark that the two-type Feller diffusions discussed in this section can be obtained as rescaling limits in space and time of two-type Galton Watson processes with immigration (see for instance \cite{Ma2}).

\subsection{The asymmetric two-island frequency process}\label{ATIFM} 
In this section, we will introduce a model for the relative frequency of individuals of type $\textbf{x}$ on two different islands under the assumption that the total mass on each island remains constant in time.
We refer to this model as the asymmetric two-island frequency process.
For $\mathbf{z}=(z_1,z_2)\in(0,\infty)^2$ and $\mathbf{r}=(r_1,r_2)\in[0,1]^2$ we define the asymmetric two-island frequency process $\mathbf{R}^{(\mathbf{z},\mathbf{r})}:=\{\mathbf{R}^{(\mathbf{z},\mathbf{r})}_t=(R^{(\mathbf{z},\mathbf{r}),1}_t,R^{(\mathbf{z},\mathbf{r}),2}_t):t\geq0\}$ as the solution of the following stochastic differential equation 
\begin{align}\label{sde}
	d\mathbf{R}^{(\mathbf{z},\mathbf{r})}_t&=b(\mathbf{R}^{(\mathbf{z},\mathbf{r})}_t)dt+\sqrt{\frac{c_1^X}{z_1}R^{(\mathbf{z},\mathbf{r}),1}_t(1-R^{(\mathbf{z},\mathbf{r}),1}_t)^2+\frac{c_1^Y}{z_1}(R^{(\mathbf{z},\mathbf{r}),1}_t)^2(1-R^{(\mathbf{z},\mathbf{r}),1}_t)}1_{\{R^{(\mathbf{z},\mathbf{r}),1}_t\in[0,1]\}}dW^{1}_te_1\notag\\
	&+\sqrt{\frac{c_2^X}{z_2}R^{(\mathbf{z},\mathbf{r}),2}_t(1-R^{(\mathbf{z},\mathbf{r}),2}_t)^2+\frac{c_2^Y}{z_2}(R^{(\mathbf{z},\mathbf{r}),2}_t)^2(1-R^{(\mathbf{z},\mathbf{r}),2}_t)}1_{\{R^{(\mathbf{z},\mathbf{r}),2}_t\in[0,1]\}}dW^{2}_te_2, \qquad t\geq 0,\notag\\
	\mathbf{R}^{(z,r)}_0&=(r_1,r_2),
\end{align}
where $\mathbf{W}:=\{\mathbf{W}_t=(W^1_t,W^2_t):t\geq0 \}$ is a two-dimensional Brownian motion.
And for $\mathbf{x}=(x_1,x_2)\in\R^2$, we have $b(\mathbf{x}):=(b^1(\mathbf{x}),b^2(\mathbf{x}))$ with 
\begin{align}\label{b_1_a}
	b^1(\mathbf{x}):&=\Bigg[(B_{11}^X-B_{11}^Y)x_1(1-x_1)1_{\{x_1\in[0,1]\}}+B^X_{12}\frac{z_2}{z_1}(1-x_1)x_21_{\{x_1\leq 1,x_2\geq0 \}}-B^Y_{12}\frac{z_2}{z_1}x_1(1-x_2)1_{\{x_1\geq 0,x_2\leq1 \}}\notag\\&-\frac{2}{z_1}(c_1^X-c_1^Y)x_1(1-x_1)1_{\{x_1\in[0,1]\}}\Bigg]+\frac{\beta_1^X}{z_1}(1-x_1)-\frac{\beta_1^Y}{z_1}x_1,
\end{align}
and
\begin{align}\label{b_2_a}
	b^2(\mathbf{x}):&=\Bigg[(B^X_{22}-B_{22}^Y)x_2(1-x_2)1_{\{x_2\in[0,1]\}}+B^X_{21}\frac{z_1}{z_2}x_1(1-x_2)1_{\{x_1\geq 0,x_2\leq1 \}}-B^Y_{21}\frac{z_1}{z_2}x_2(1-x_1)1_{\{x_1\leq 1,x_2\geq0 \}}\notag\\&-\frac{2}{z_2}(c_2^X-c_2^Y)x_2(1-x_2)1_{\{x_2\in[0,1]\}}\Bigg]+\frac{\beta_2^X}{z_2}(1-x_2)-\frac{\beta_2^Y}{z_2}x_2.
\end{align}
The existence and uniqueness of a strong solution to \eqref{sde} will be provided in Proposition \ref{exis_uni_sde}.

We now sketch the construction of the asymmetric two-island frequency process,       $\mathbf{R}^{(\mathbf{z},\mathbf{r})}$. 
 As introduced in Section \ref{preliminiaries}, let $\mathbf{X}_t=(X_t^{1},X_t^{2})$ and $\mathbf{Y}_t=(Y_t^{1},Y_t^{2})$ be  two-dimensional Feller diffusions with continuous immigration that model the growth of the population size of type $\textbf{x}$ and $\textbf{y}$ individuals, respectively.  We will show in Section \ref{frequency_proc} that $\mathbf{R}^{(\mathbf{z},\mathbf{r})}$  can be constructed from the relative frequency processes of type $\mathbf{x}$ in each island, i.e. $\left(\frac{X_t^{1}}{X_t^{1}+Y_t^{1}}, \frac{X_t^{2}}{X_t^{2}+Y_t^{2}}\right)$, by means of the culling procedure. 

Each of the parameters can be traced to those of the associated two-dimensional Feller diffusions with immigration, and we will explain the relation between the parameters in Table \ref{table}. If the change of state parameters (i.e. $B^X_{ij}$, $B^Y_{ij}$ with $i\not=j$) and the variances $c^X_i,c^Y_i$ are different from zero only for $i,j=1,2$, we obtain a structured version of the Gillespie--Wright-Fisher diffusion introduced in \cite{Gill73,Gill74}. If $c_i^X=c_i^Y$ for $i=1,2$, we recover the two-island diffusion. Furthermore, if $c_1^X=c_1^Y$ and $c_2^X=c_2^Y=0$, the process specializes to the seedbank diffusion.

The main steps to prove that $\mathbf{R}^{(\mathbf{z},\mathbf{r})}$ arises from the relative frequency associated to two independent two-dimensional Feller diffusions with immigration as a limit when the sampling times become dense in the culling procedure are: first we describe the dynamics of the relative frequency process of type-$\textbf{x}$ individuals (when the total size of each island is allowed to vary in time) as the solution to a martingale problem;  we next show that the SDE associated to the process $\mathbf{R}^{(\mathbf{z},\mathbf{r})}$ has a unique strong solution, which is a Feller process. The previous steps allow us to reduce the convergence of the culling procedure to a mere convergence of infinitesimal generators, which is the third step.

\subsection{Moment duality for the asymmetric two-island frequency process}\label{duality-intro}
In Section \ref{duality} we show, under some conditions on the parameters of $\mathbf{R}^{(\mathbf{z},\mathbf{r})}$, that there exists a two-dimensional continuous-time Markov chain $\mathbf{N}^{(\mathbf{z},\mathbf{n})}:=\{\mathbf{N}^{(\mathbf{z},\mathbf{n})}_t=(N^{(\mathbf{z},\mathbf{n}),1}_t,N^{(\mathbf{z},\mathbf{n}),2}_t): t\geq0\}$ such that, for every $t>0$ and for every initial condition in their respective state spaces $\mathbf{n}=(n_1,n_2)$ and $\mathbf{r}=(r_1,r_2)$, the following relation holds:
$$
\E[r_1^{\mathbf{N}_t^{(\mathbf{z},\mathbf{n}),1}}r_2^{\mathbf{N}_t^{(\mathbf{z},\mathbf{n}),2}}|\mathbf{N}^{(\mathbf{z},\mathbf{n})}_0=(n_1,n_2)]=\E[(\mathbf{R}^{(\mathbf{z},\mathbf{r}),1})^{n_1}(\mathbf{R}^{(\mathbf{z},\mathbf{r}),2})^{n_2}|\mathbf{R}^{(\mathbf{z},\mathbf{r})}_0=(r_1,r_2)].
$$
For each $(i_1,i_2),(j_1,j_2)\in\mathbb{N}_0^2\cup\{\Delta\}$, and $\mathbf{z}=(z_1,z_2)\in(0,\infty)^2$ we define the following set of real numbers, which should be thought as transition rates
\begin{equation}
	q_{(i_1,i_2),(j_1,j_2)}^{\mathbf{z}}=
	\begin{cases}
		\displaystyle i_1\frac{\beta_1^Y}{z_1}+i_2\frac{\beta_2^Y}{z_2}&\mbox{if } (j_1,j_2)=\Delta,\\
		\displaystyle i_1s_1+i_2\frac{z_1}{z_2}(B_{21}^Y-B_{21}^X)+\frac{2}{z_1}\binom{i_1}{2}(c_1^X-c_1^Y)&\mbox{if } (j_1,j_2)=(i_1+1,i_2),\\
		\displaystyle i_2s_2+i_1\frac{z_2}{z_1}(B_{12}^Y-B_{12}^X)+\frac{2}{z_2}\binom{i_2}{2}(c_2^X-c_2^Y)&\mbox{if } (j_1,j_2)=(i_1,i_2+1),\\
		i_1B_{12}^X\displaystyle\frac{z_2}{z_1}&\mbox{if } (j_1,j_2)=(i_1-1,i_2+1),\\
		i_2B_{21}^X\displaystyle\frac{z_1}{z_2}&\mbox{if } (j_1,j_2)=(i_1+1,i_2-1),\\
		\displaystyle i_1\frac{\beta_1^X}{z_1}+\frac{2}{z_1}\binom{i_1}{2}c_1^X&\mbox{if } (j_1,j_2)=(i_1-1,i_2),\\
		\displaystyle i_2\frac{\beta_2^X}{z_2}+\frac{2}{z_2}\binom{i_2}{2}c_2^X&\mbox{if } (j_1,j_2)=(i_1,i_2-1).
	\end{cases}
\end{equation}
where $s_i=(B_{ii}^Y-B_{ii}^X)+\frac{2}{z_i}(c_i^X-c_i^Y)$ for $i=1,2$.

If $q^{\mathbf{z}}_{(i_1,i_2),(j_1,j_2)}\geq0$ for all $(i_1,i_2),(j_1,j_2)\in\mathbb{N}^2_0\cup\{\Delta\}$ the continuous-time Markov chain $\mathbf{N}^{(\mathbf{z},\mathbf{n})}$ has state space $\mathbb{N}_0^2\cup\{\Delta\}$, such that $\mathbf{N}^{(\mathbf{z},\mathbf{n})}_0=\mathbf{n}$ with $\mathbf{n}=(n_1,n_2)\in\mathbb{N}^2_0$ and generator given by $\mathcal{Q}^{\mathbf{z}}=\{q^{\mathbf{z}}_{(i_1,i_2),(j_1,j_2)}:(i_1,i_2),(j_1,j_2)\in\mathbb{N}_0^2\cup\{\Delta\}\}$.

This relation between the moments of $\mathbf{R}^{(\mathbf{z},\mathbf{r})}$ and the probability generating function of $\mathbf{N}^{(\mathbf{z},\mathbf{n})}$ is known as moment duality and constitutes a fundamental tool in population genetics. The most important example of this relation is the duality between the Wright--Fisher diffusion and the Kingman coalescent. Moment duality formally links the evolution of a population with its ancestry. 
There are several other examples, for instance the celebrated ancestral selection graph \cite{KN,KNG} gives a notion of extended ancestry for the Wright--Fisher diffusion with selection via moment duality. Additionally, duality has been used for the two-island model and the seedbank model in the presence of mutation in \cite{BBGW19}. In \cite{GPP}, the duality between the Gillespie--Wright--Fisher diffusion and a coalescing/pairwise-branching process was established and interpreted in terms of efficiency in \cite{GMP}.

The existence of a moment dual for the asymmetric two-island frequency process generalizes all the results in the previous paragraph. An interesting fact about this new duality result is that in the dual process the particles can produce offspring of both types. 

In Table \ref{table} we present a \textit{dictionary} that translates the different parameters of the Feller diffusions with immigration, into the parameters of the asymmetric two-island frequency processes and the parameters of its dual. This is useful to understand how the three families of processes are related. For the last column, we remark that the parameters of the moment dual of the asymmetric two-island frequency process must fulfill some conditions to ensure the existence of the process (see Theorem \ref{dual}). 

In the first two lines of Table \ref{table}, we see that the differences between the Malthusians associated to the Feller diffusions with immigration $\left((B_{ii}^X-B_{ii}^Y),\ i=1,2\right)$, become the coefficients of logistic drift terms in the asymmetric two-island frequency process associated to selection, and in the moment dual they become branching rates.

In lines 3-6 of Table \ref{table}, we can observe that the variances of each of the Feller diffusions have three different roles in the associated frequency process. The variances of the process $\mathbf{X}$ ($c_i^X$, $i=1,2$) dictate the strength of the random genetic drift for each island. The differences of the variances ($(c_i^Y-c_i^X)$, $i=1,2$), which are assumed to satisfy the condition $c_i^Y-c_i^X\leq 0$ for $i=1,2$, in order for the moment dual to exist, represent a frequency-dependent random genetic drift factor, and they also contribute with extra logistic drift terms in the asymmetric two-island frequency process. As expected, the random genetic drift is dual to Kingman's coalescence. Furthermore, for the dual process, the differences $c_i^X-c_i^Y$, $i=1,2$ determine the rate of pairwise branching and give an extra (individual) branching rate.

In lines 7-10, we consider the immigration terms of the Feller diffusion. These become drift terms that take each of the components of the asymmetric two-island frequency process away from its boundary points $\{0,1\}$. In the dual, they become death rates (the rate at which an individual is removed) and killing rates (the rate at which the process is absorbed in a cemetery state). 

Finally, in the last two lines, we have the parameters that govern the cross production of individuals, i.e. the way that individuals in island 1 produce individuals on island 2 and the other way around. In the frequency process, $\frac{z_2}{z_1}B_{12}^X$ ponderates a drift towards the identity, while the term $\frac{z_2}{z_1}(B_{12}^X -B_{12}^Y)$ represents a novel drift term that resembles selection but that depends on the frequency in both islands instead. In the moment dual the term $\frac{z_2}{z_1}B_{12}^X$ becomes the migration rate (i.e. the rate at which an individual goes from island 1 to island 2) and $\frac{z_2}{z_1}(B_{12}^Y-B_{12}^X)$ becomes a branching/migration term that determines the rate at which individuals in island 1 branch and allocate the newborn individual in island 2. The roles of $B_{21}^X$ and $B_{21}^Y$ are similar.

\begin{table}
	\begin{tabular}{ | m{1em} | m{3cm}| m{9.5cm} | m{3cm}| }
		\hline
		& Feller Diffusion with immigration & Asymmetric two-island frequency process & Dual process \\
		\hline
		\multirow{1}{5em}{1} & $B_{11}^X,B_{11}^X$ & $(B_{11}^X-B_{11}^Y)R^{({\mathbf{z}},{\mathbf{r}}),1}_t(1-R^{({\mathbf{z}},{\mathbf{r}}),1}_t)$ & $i_1(B_{11}^Y-B_{11}^X)$\\
		\hline
		\multirow{1}{5em}{2} & $B_{22}^X,B_{22}^X$ & $(B_{22}^X-B_{22}^Y)R^{({\mathbf{z}},{\mathbf{r}}),2}_t(1-R^{({\mathbf{z}},{\mathbf{r}}),2}_t)$ & $i_2(B_{22}^Y-B_{22}^X)$\\
		\hline
		\multirow{2}{5em}{3} & $c_{1}^X,c_{1}^Y$ & $\sqrt{\frac{c_1^X}{z_1}R^{({\mathbf{z}},{\mathbf{r}}),1}_t(1-R^{({\mathbf{z}},{\mathbf{r}}),1}_t)+\frac{(c_1^Y-c_1^X)}{z_1}(R^{({\mathbf{z}},{\mathbf{r}}),1}_t)^2(1-R^{({\mathbf{z}},{\mathbf{r}}),1}_t)}$ & $\frac{2}{z_1}\binom{i_1}{2}(c_1^X-c_1^Y)$\\
		& & & $\frac{2}{z_1}\binom{i_1}{2}c_1^X$\\
		\hline
		\multirow{1}{5em}{4} & $c_{1}^X,c_{1}^Y$ & $-2z_1^{-1}(c_1^X-c_1^Y)R^{({\mathbf{z}},{\mathbf{r}}),1}_t(1-R^{({\mathbf{z}},{\mathbf{r}}),1}_t)$ & $i_12z_1^{-1}(c_1^X-c_1^Y)$\\
		\hline
		\multirow{2}{5em}{5} & $c_{2}^X,c_{2}^Y$ & $\sqrt{\frac{c_2^X}{z_2}R^{({\mathbf{z}},{\mathbf{r}}),2}_t(1-R^{({\mathbf{z}},{\mathbf{r}}),2}_t)+\frac{(c_2^Y-c_2^X)}{z_2}(R^{({\mathbf{z}},{\mathbf{r}}),2}_t)^2(1-R^{({\mathbf{z}},{\mathbf{r}}),2}_t)}$ & $\frac{2}{z_2}\binom{i_2}{2}(c_2^X-c_2^Y)$\\
		& & & $\frac{2}{z_2}\binom{i_2}{2}c_2^X$\\
		\hline
		\multirow{1}{5em}{6} &$c_{2}^X,c_{2}^Y$ & $-2z_2^{-1}(c_2^X-c_2^Y)R^{({\mathbf{z}},{\mathbf{r}}),2}_t(1-R^{({\mathbf{z}},{\mathbf{r}}),2}_t)$ & $i_22z_2^{-1}(c_2^X-c_2^Y)$\\
		\hline
		\multirow{1}{5em}{7} & $\beta_1^X$ & $\beta_1^Xz_1^{-1}(1-R^{({\mathbf{z}},{\mathbf{r}}),1}_t)$ & $i_1z_1^{-1}\beta_1^X$\\
		\hline
		\multirow{1}{5em}{8} & $\beta_2^X$ & $\beta_2^Xz_2^{-1}(1-R^{({\mathbf{z}},{\mathbf{r}}),2}_t)$ & $i_2z_2^{-1}\beta_2^X$\\
		\hline
		\multirow{1}{5em}{9} & $\beta_1^Y$ & $-\beta_1^Yz_1^{-1}R^{({\mathbf{z}},{\mathbf{r}}),1}_t$ & $i_1z_1^{-1}\beta_1^Y$\\
		\hline
		\multirow{1}{5em}{10} & $\beta_2^Y$ & $-\beta_2^Yz_2^{-1}R^{({\mathbf{z}},{\mathbf{r}}),2}_t$ & $i_1z_2^{-1}\beta_2^Y$\\
		\hline
		\multirow{2}{5em}{11} & $B_{12}^X,B_{12}^Y$ & $\frac{z_2}{z_1}(B_{12}^X-B_{12}^Y)R^{({\mathbf{z}},{\mathbf{r}}),1}_t(1-R^{({\mathbf{z}},{\mathbf{r}}),2}_t)$ & $i_1\frac{z_2}{z_1}(B_{12}^Y-B_{12}^X)$\\
		& & $\frac{z_2}{z_1}B_{12}^X(R^{({\mathbf{z}},{\mathbf{r}}),2}_t-R^{({\mathbf{z}},{\mathbf{r}}),1}_t)$ &  $i_1\frac{z_2}{z_1}B_{12}^X$\\
		\hline
		\multirow{2}{5em}{12} & $B_{21}^X,B_{21}^Y$ & $\frac{z_1}{z_2}(B_{21}^X-B_{21}^Y)R^{({\mathbf{z}},{\mathbf{r}}),2}_t(1-R^{({\mathbf{z}},{\mathbf{r}}),1}_t)$ &  $i_2\frac{z_1}{z_2}(B_{21}^Y-B_{21}^X)$\\
		& & $\frac{z_1}{z_2}B_{21}^X(R^{({\mathbf{z}},{\mathbf{r}}),1}_t-R^{({\mathbf{z}},{\mathbf{r}}),2}_t)$ & $i_2\frac{z_1}{z_2}B_{21}^X$\\
		\hline
	\end{tabular}
	\caption{Relation between the parameters of the Feller diffusions with immigration and those of their associated asymmetric two-island frequency process as well as those of its moment dual.}
	\label{table}
\end{table}
\subsection{The large population deterministic limit}

To understand the expected long-time behaviour of the the asymmetric two-island frequency process, in Section \ref{asymp}, we will send the total mass of the population in each island, fixed at each of the sampling times, to infinity and obtain a system of ordinary differential equations. 
The resulting ODE is given by
\begin{align}\label{det_sys_0}
	\frac{d}{dt}r^{1}_t&=B^X_{12}(r^{2}_t-r^{1}_t)-(B^X_{12}-B^Y_{12})(r^{1}_tr^{2}_t-r^{1}_t)+(B_{11}^X-B_{11}^Y)r^{1}_t(1-r^{1}_t),\notag\\
	\frac{d}{dt}r^{2}_t&=B^X_{21}(r^{1}_t-r^{2}_t)-(B^X_{21}-B^Y_{21})(r^{1}_tr^{2}_t-r^{2}_t)+(B_{22}^X-B_{22}^Y)r^{2}_t(1-r^{2}_t), \qquad \text{ $t\geq0$}.
\end{align}
The parameter $B^X_{ij}$ is the continuous state equivalent to the rate, in the asssociated Feller diffusion with immigration $\textbf{X}$, at which type $\textbf{x}$ individuals in island $j$ produce offspring in island $i$ of type $\textbf{x}$. A negative $B^X_{ij}$ is only posible if $i=j$ and this is related to a tendency to produce less than one individual of type $\textbf{x}$. The description of the parameters $B^Y_{ij}$ associated with the process $\textbf{Y}$ is analogous.

The first term on each line pushes the solutions towards the identity line. The other terms depend on the differences between the reproduction rates of each type. Each of them pushes the solution towards one of the boundary points $\{(0,0), (1,1)\}$ which are associated with the event of fixation of one of the types in the population. However, at each point of the unit square, all the forces act simultaneously and complex behaviour arises. To predict which type of individuals go to fixation one needs to consider all of the parameters, and it is even possible that the system of ODEs points to an internal point, implying that this complex interaction can favor coexistence.

To visualize the complexity of the interaction between these parameters we invite the reader to visit the web page of Julio Nava (https://julionava.shinyapps.io/app$\_$prueba/), where it is possible to find a user-friendly interface that solves the system of ODEs numerically and presents different plots of interest. 

\subsection{The benefits of sustaining a seedbank}
Recall the interpretation of the islands $1$ and $2$ as active and dormant. In this section we study a toy model that gives some insight on the advantage for populations sustaining seedbanks. Here, $\textbf{x}$ sustains a seed bank, while $\textbf{y}$ does not. To be precise, in Equation \eqref{det_sys_0} we let $B_{12}^Y=B_{21}^Y=B_{22}^Y=0$ and  $B_{22}^X=-B_{12}^X$. Then, the system reduces to
\begin{align}\label{toymodel}
	\frac{d}{dt}r^{1}_t&=B^X_{12}r^{2}_t(1-r^{1}_t)+(B_{11}^X-B_{11}^Y)r^{1}_t(1-r^{1}_t),\notag\\
	\frac{d}{dt}r^{2}_t&=B^X_{21}r^{1}_t(1-r^{2}_t)-B^X_{12}r^{2}_t(1-r^{2}_t), \qquad \text{ $t\geq0$}.
\end{align}
Observe that the solution to the second equation is given by
\begin{align*}
	r^{2}_t&=1-(1-r^2_0)e^{-\int_0^t\left[B^X_{21}r^{1}_s-B^X_{12}r_s^2\right]ds}, \qquad \text{ $t\geq0$}.
\end{align*}
Now, as only type $\textbf{x}$ individuals produce seeds it is reasonable to asume that the frequency of $\textbf{x}$-type seeds at time zero is unity i.e. $r^2_0=1$. Note that this implies that 
\begin{equation}\label{r_t_1}
	r^2_t=1, \qquad\text{for all $t\geq 0$}.
\end{equation}
So if $r_0^2=1$, by substituting \eqref{r_t_1} in \eqref{toymodel} we obtain
\[
\frac{d}{dt}r^{1}_t=B^X_{12}(1-r^{1}_t)+(B_{11}^X-B_{11}^Y)r^{1}_t(1-r^{1}_t),\qquad t>0,
\]
which has a solution given by
\begin{align*}
	r^1_t=\frac{e^{(A+B)t}-CA}{CB+e^{(A+B)t}}, \qquad t\geq0.
\end{align*}
with $A=B^X_{12}$, $B=B_{11}^X-B_{11}^Y$ and $C=\frac{1-r_0^1}{A+Br_0^1}$.

It is now easy to see that if $B_{12}^X+B_{11}^X\geq B_{11}^Y$, then type-$\textbf{x}$ individuals have a strong selective advantage, meaning that
\[
\lim_{t\to\infty}r^1_t=1.
\]
On the other hand, if $B_{12}^X+B_{11}^X<B_{11}^Y$, then
\[
\lim_{t\to\infty}r^1_t=\frac{B^X_{12}}{(B_{11}^Y-B_{11}^X)}\in(0,1).
\]
Note that $B_{12}^X+B_{11}^X$ is the rate at which type $\mathbf{x}$ individuals produce offspring and seeds, of type $\mathbf{x}$. For this reason, we can say that if the population sustaining a seedbank reproduces faster, it will have a strong selective advantage, while if they are slower, the populations will coexist.  In other words, we observe that the seed producers never have a strong selective disadvantage: coexistence is their worst case scenario. 

	\section{Relative frequency and total population size}\label{pop_proc}
	As discussed in the previous section, we will distinguish between two different types $\textbf{x}$ and $\textbf{y}$ in the population distributed between two islands $1$ and $2$. The evolution of the total mass of individuals of type $\textbf{x}$ and $\textbf{y}$ in the island $i$ will be described by the Feller diffusions with immigration $X^i$ and $Y^i$ respectively, for $i=1,2$.
	
	We will now characterize the relative frequency of individuals of type $\textbf{x}$ and the total size of the population associated with each of the different islands. To this end, for each island $i=1,2$, let us consider the relative frequency process of individuals of type $\textbf{x}$ given by
	\begin{align}\label{rel_freq}
		R^i_t:=\frac{X^i_t}{X^i_t+Y^i_t}1_{\{t\leq  \tau^i\}}+R^i_{\tau} 1_{\{t>  \tau^i\}}
		\qquad\text{$t\geq0$},\qquad
		R_0^i=r^i,
	\end{align}
where $\tau^i=\inf\{t\geq0: X^{i}Y^{i}=0\}$ and $r^i=x^{i}/(x^{i}+y^{i})$.

	To obtain a complete description of the dynamics of the population we will also consider, for each $i=1,2$, the total population size of the island $i$ defined by
	\begin{align}\label{pop_size}
		Z^i_t:=X^i_t+Y^i_t,\qquad t\geq0, \qquad Z^i_0=z^i,
	\end{align}
	where $z^i=x^i+y^i$.
	
	In the next result we will obtain the dynamics of the process $(\mathbf{R},\mathbf{Z}):=(R^1,R^2,Z^1,Z^2)$ defined by \eqref{rel_freq} and \eqref{pop_size}. We note that the process $(\mathbf{R},\mathbf{Z})$ is a Markov process, and hence we will describe its evolution through the solution to a martingale problem in the next result. We denote the law of the process $(\mathbf{R},\mathbf{Z})$ starting from the initial position $(\mathbf{r},\mathbf{z})=(r^1,r^2,z^1,z^2)\in[0,1]^2\times(0,\infty)^2$ by $\mathbb{P}_{(\mathbf{r},\mathbf{z})}$. Accordingly, we denote the associated expectation operator by $\E_{(r,z)}$.
	
	For $\varepsilon\in(0,z^1\wedge z^2)$ and $L>|\mathbf{z}|$ we denote by $\tau_{\varepsilon}:=\inf\{t>0:Z^1_t\wedge Z^2_t=\varepsilon\}$ and $\tau_L:=\inf\{t>0:Z^1_t\vee Z^2_t=L\}$ to the first hitting times of the level $\varepsilon$ (resp. $L$) for the process $Z^1\wedge Z^2$ (resp. $Z^1\vee Z^2$). We remark that the process $(\mathbf{R},\mathbf{Z})$ stopped at the stopping time $\tau=\tau_{\varepsilon}\wedge\tau_L$ completely encodes the dynamics of the two subpopulations of the types $\textbf{x},\textbf{y}$ in the different islands, originally described by the processes $\mathbf{X}$ and $\mathbf{Y}$, before the size of the population in any of the islands becomes relatively small or large.
	
	The next result describes the dynamics of the process $(\mathbf{R},\mathbf{Z})$ until the stopping time $\tau$ as the solution to a martingale problem. We defer its proof to Appendix \ref{mart_prob_proof}.
	
	\begin{lemma}\label{mart_prob}
		For any $f\in\mathcal{C}^2([0,1]^2\times\R_+^2)$ the process
		\begin{align*}
			f(R^1_{t\wedge\tau},R^2_{t\wedge\tau},Z^1_{t\wedge\tau},Z^2_{t\wedge\tau})-f(r^1,r^2,z^1,z^2)-\int_0^{t\wedge\tau}\mathcal{L}f(R^1_s,R^2_s,Z^1_s,Z^2_s)ds,\qquad t\geq 0,
		\end{align*}
		is a local martingale where
		\begin{align}\label{gen_1}
			\mathcal{L}f(r_1,r_2,z_1,z_2)=\sum_{i=1}^3B^i(r_1,r_2,z_1,z_2),
		\end{align}
	and $B^i$ is given in \eqref{b_1}-\eqref{b_3} for $i=1,\dots,3$. 
\end{lemma}
	\section{Asymmetric two-island frequency process and culling of general two-island models.}\label{frequency_proc}
	In the previous section, we described the dynamics of the process $(\mathbf{R},\mathbf{Z})=(R^1,R^2,Z^1,Z^2)$, stopped at the stopping time $\tau$, as the solution to a martingale problem. 
	In this section, we will construct the asymmetric frequency process $\mathbf{R}^{(\mathbf{z},\mathbf{r})}$, introduced in Section \ref{ATIFM},  
	from the process $(R^1,R^2,Z^1,Z^2)$ by a sampling/culling procedure similar to the one introduced in \cite{CGP}.
	
	We start this section by showing that the process $\mathbf{R}^{(z,r)}$ is well-defined as the unique strong solution to \eqref{sde}, the proof of this result is defered to Appendix \ref{unique}.
	\begin{proposition}\label{exis_uni_sde}
		There exists a unique strong solution $\mathbf{R}^{(\mathbf{z},\mathbf{r})}$ to \eqref{sde} such that $\mathbf{R}^{(\mathbf{z},\mathbf{r})}\in[0,1]^2$ for all $t\geq0$ $\mathbb{P}$-a.s. Furthermore, for any $t>0$ there exists $C(t)>0$ such that 
		\begin{align}\label{est_1}
			\E\left[\|\mathbf{R}^{(\mathbf{z},\mathbf{r})}_t-\mathbf{R}^{(\mathbf{z},\overline{\mathbf{r}})}_t\|\right]\leq C(t)\|\mathbf{r}-\overline{\mathbf{r}}\|,\qquad \mathbf{r},\overline{\mathbf{r}}\in[0,1]^2.
		\end{align}
	\end{proposition}
	\begin{remark}
	We note that in the particular case in which the Feller diffusions with immigration $\mathbf{X}$ and $\mathbf{Y}$ have the same distribution then, for $t>0$, the SDE given in \eqref{sde} takes the following form 
	\begin{align*}
		dR^{(z,r),1}_t&=\left[B_{12}^X\frac{z_2}{z_1}(R^{(z,r),2}_t-R^{(z,r),1}_t)+\frac{\beta_1^X}{z_1}(1-R^{(z,r),1}_t)-\frac{\beta_1^X}{z_1}R^{(z,r),1}_t\right]dt+\sqrt{\frac{c_1^X}{z_1}R^{(z,r),1}_t(1-R^{(z,r),1}_t)}dW^1_t\\
		dR^{(z,r),2}_t&=\left[B_{21}^X\frac{z_1}{z_2}(R^{(z,r),1}_t-R^{(z,r),2}_t)+\frac{\beta_2^X}{z_2}(1-R^{(z,r),2}_t)-\frac{\beta_2^X}{z_2}R^{(z,r),2}_t\right]dt+\sqrt{\frac{c_2^X}{z_2}R^{(z,r),2}_t(1-R^{(z,r),2}_t)}dW^2_t,
	\end{align*}
	with $\mathbf{R}^{(z,r)}_0=(r_1,r_2)$.
	And hence, we recover the general two-island model as in Remark 1.2 of \cite{BBGW19}.
\end{remark}
	In the next result we show that the process $\mathbf{R}^{(\mathbf{z},\mathbf{r})}$ is Feller and we characterize the process by computing its infinitesimal generator, we defer the proof to Appendix \ref{generator}.
	\begin{proposition}\label{gen_culling}
		For any $\mathbf{z}=(z_1,z_2)\in(0,\infty)^2$, $\mathbf{R}^{(\mathbf{z},\mathbf{r})}$ is a Feller process and its infinitesimal generator is given for any $f\in\mathcal{C}^2([0,1]^2)$ by
		\begin{align}\label{fun_d}
			\mathcal{L}^{\mathbf{z}}f(r_1,r_2):=\sum_{i=1}^2D^{i}_f(r_1,r_2,z_1,z_2),\qquad \text{for $\mathbf{r}=(r_1,r_2)\in[0,1]^2$},
		\end{align}
		where 
		\begin{align*}
			D^1_f(r_1,&r_2,z_1,z_2)\notag\\:&=\partial_1f(\mathbf{r})(1-r_1)\left(B_{11}^Xr_1+B^X_{12}\frac{z_2}{z_1}r_2+\frac{\beta_1^X}{z_1}\right)-\partial_1f(\mathbf{r})r_1\left(B_{11}^Y(1-r_1)+B^Y_{12}\frac{z_2}{z_1}(1-r_2)+\frac{\beta_1^Y}{z_1}\right)\notag\\
			&+\partial_2f(\mathbf{r})(1-r_2)\left(B_{21}^X\frac{z_1}{z_2}r_1+B^X_{22}r_2+\frac{\beta_2^X}{z_2}\right)-\partial_2f(\mathbf{r})r_2\left(B_{21}^Y\frac{z_1}{z_2}(1-r_1)+B^Y_{22}(1-r_2)+\frac{\beta_2^Y}{z_2}\right),\notag\\
			D^2_f(r_1,&r_2,z_1,z_2)\notag\\:&=\frac{c_1^X}{z_1}\left(r_1(1-r_1)^2\partial^2_{11}f(\mathbf{r})-2r_1(1-r_1)\partial_1f(\mathbf{r})\right)+\frac{c_2^X}{z_2}\left(r_2(1-r_2)^2\partial^2_{22}f(\mathbf{r})-2r_2(1-r_2)\partial_2f(\mathbf{r})\right)\notag\\
			&+\frac{c_1^Y}{z_1}\left(r_1^2(1-r_1)\partial^2_{11}f(\mathbf{r})+2r_1(1-r_1)\partial_1f(\mathbf{r})\right)+\frac{c_2^Y}{z_2}\left(r_2^2(1-r_2)\partial^2_{22}f(\mathbf{r})+2r_2(1-r_2)\partial_2f(\mathbf{r})\right).
		\end{align*}
	\end{proposition}
	\subsection{Culling of the population process}
	In Section \ref{pop_proc} we obtained the evolution of the process \linebreak $(\mathbf{R},\mathbf{Z})=(R^1,R^2,Z^1,Z^2)$ as the solution to a martingale problem. We recall that for each $i=1,2$, the component $R^i$ represents the frequency of individuals of type $\textbf{x}$ in the $i$-th island, while $Z^i$ describes its total population size. 
	
	As in many population models we are interested in maintaining the total size of the population constant while allowing the frequency processes of individuals of type $\textbf{x}$ associated to each island, $(R^1,R^2)$, to evolve randomly; obtaining a two-dimensional stochastic process. Inspired by the work of Gillespie \cite{Gill73} (see also Section 4.2 in \cite{CGP}) we will use a sampling method called $\textit{culling}$, in which at each sampling time we take the total size of the population in each island back to a constant level. This method will allow us to obtain a random description of the frequency processes under the assumption that the total size of the population in each island is constant.
	
	Let us consider fixed population sizes $z_1,z_2>0$ for the first and second islands, respectively. For $\mathbf{z}=(z_1,z_2)$,  consider a sequence of pseudo-Poisson processes $\left\{ \overline{\mathbf{R}}^{(\mathbf{z},n)}_t=\left(\overline{R}^{(\mathbf{z},n),1}_t,\overline{R}^{(\mathbf{z},n),2}_t\right): t\geq0\right\}_{n\geq 1}$ (see Chapter 10 in \cite{Ka} for a detailed introduction to pseudo-Poisson processes). We denote the law 
	of the process $\overline{\mathbf{R}}^{(\mathbf{z},n)}$ by $\mathbf{P}_{\mathbf{r}}$ when started at the state $\mathbf{r}=(r_1,r_2)\in[0,1]^2$. For each fixed $n\geq1$, the Markov process $\overline{\mathbf{R}}^{(\mathbf{z},n)}$ has jump times $(T^{m,n})_{m\geq1}$ given by independent exponential random variables with rate $n$, and transition kernel $\{\kappa^{(z,n)}:n\geq 1\}$ defined for $\mathbf{y}\in[0,1]^2$, and $A\in\mathcal{B}([0,1]^2)$ by
	\begin{align}\label{kernel_cull}
		\kappa^{(z,n)}(\mathbf{y},A)=\mathbf{P}_{\mathbf{y}}\left(\overline{R}^{(z,n)}_{T^{1,n}}\in A\right)=\mathbb{P}_{(y_1,y_2,z_1,z_2)}\left(\left(R^{1}_{\frac{1}{n}\wedge\tau},R^{2}_{\frac{1}{n}\wedge\tau},Z^{1}_{\frac{1}{n}\wedge\tau},Z^{2}_{\frac{1}{n}\wedge\tau}\right)\in A\times\R_+^2\right),
	\end{align}
	where $\tau=\tau_{\varepsilon}\wedge\tau_L$ is defined in Section \ref{pop_proc}.
	
	By the proof of Proposition 17.28 in \cite{Ka} the process $\overline{\mathbf{R}}^{(\mathbf{z},n)}$ is Feller and has an infinitesimal generator given for any $f\in\mathcal{C}([0,1]^2)$ by
	\begin{align}\label{discrete_gen}
		\overline{\mathcal{L}}^{(\mathbf{z},n)}f(\mathbf{r})=n\int_{[0,1]^2}(f(\mathbf{u})-f(\mathbf{r}))\kappa^{(z,n)}(\mathbf{r},\ud\mathbf{u}),\qquad \mathbf{r}\in[0,1]^2.
	\end{align}
	
	For each $n\geq 1$, we can think of $\overline{\mathbf{R}}^{(\mathbf{z},n)}_{T^{1,n}}$ as a sampling of the first two coordinates of the process \linebreak $(R^1_{t\wedge\tau},R^2_{t\wedge\tau},Z^1_{t\wedge\tau},Z^2_{t\wedge\tau})_{t\geq0}$ started at the position $(\mathbf{r},\mathbf{z})\in[0,1]^2\times\R_+^2$ at time $t=\frac{1}{n}$. Then, in order to define $\overline{\mathbf{R}}^{(\mathbf{z},n)}_{T^{2,n}}$, we use the fact that the process $(R^1,R^2,Z^1,Z^2)$ is an homogenouos Markov process, to restart the process $(R^1_{t\wedge\tau},R^2_{t\wedge\tau},Z^1_{t\wedge\tau},Z^2_{t\wedge\tau})_{t\geq0}$ at the initial position $(R^1_{\frac{1}{n}\wedge \tau},R^2_{\frac{1}{n}\wedge \tau},z_1,z_2)$ and sample the process at time $t=\frac{1}{n}$. By proceeding with this algorithm we can define the process $\overline{\mathbf{R}}^{(\mathbf{z},n)}_t$ for any $t>0$.
	
	Under this construction the process $\overline{\mathbf{R}}^{(\mathbf{z},n)}$ evolves in the time interval $[T^{m,n},T^{m+1,n})$  as the first two coordinates of the process $(R^1,R^2,Z^1,Z^2)$, on the interval $[0,\frac{1}{n}\wedge\tau)$, started from the initial position \linebreak $(\overline{\mathbf{R}}^{(\mathbf{z},n),1}_{T^{m,n}},\overline{\mathbf{R}}^{(\mathbf{z},n),2}_{T^{m,n}},z_1,z_2)$. Therefore, the process $\overline{\mathbf{R}}^{(\mathbf{z},n)}$ behaves like the first two coordinates of the process $(R^1,R^2,Z^1,Z^2)$ but with the fluctuations of the total size of the population at each island $(Z^1,Z^2)$, around $(z_1,z_2)$, becoming smaller as $n\to\infty$, i.e. as the sampling times converge to zero.
	
	We denote by $\mathbb{D}([0,T],[0,1]^2)$ the space of functions with cadlag paths on $[0,T]$ taking values on $[0,1]^2$ endowed with the Skorohod topology.
	In the next result we show that the sequence of processes $\{(\overline{\mathbf{R}}^{(\mathbf{z},n)})_{t\geq0}:n\geq 1\}$ converges weakly in the space $\mathbb{D}([0,T],[0,1]^2)$ for any $T>0$, to the process $\mathbf{R}^{(\mathbf{z},\mathbf{r})}$, which is the unique solution to the stochastic differential equation given in \eqref{sde}. By this convergence and the previous construction we can understand the process $\mathbf{R}^{(\mathbf{z},\mathbf{r})}$ as describing the evolution of the frequency of individuals of type $\textbf{x}$ in each island but under the assumption that the total size of the population in the $i$-th island is constant and given by $z_i>0$ for each $i=1,2$.
	
	\begin{theorem}
		Assume that $\overline{\mathbf{R}}^{(\mathbf{z},n)}$ is started from $\mathbf{r}\in[0,1]^2$ for every $n\geq1$. Then, for any fixed $\mathbf{z}=(z_1,z_2)\in(0,\infty)^2$ and $T>0$, $\overline{\mathbf{R}}^{(\mathbf{z},n)}\to \mathbf{R}^{(\mathbf{z},\mathbf{r})}$ as $n\to\infty$ weakly in $\mathbb{D}([0,T],[0,1]^2)$.
	\end{theorem}
	\begin{proof}
		We note that for each $n\geq 1$, the process $\overline{\mathbf{R}}^{(\mathbf{z},n)}$ takes values in the set $[0,1]^2$, which is compact, and the same holds true for the process $\mathbf{R}^{(\mathbf{z},\mathbf{r})}$, then following Theorem 17.28 in \cite{Ka} we only need to verify that for any $f\in\mathcal{C}^2([0,1]^2)$	
		\begin{align*}
			\overline{\mathcal{L}}^{(\mathbf{z},n)}f\to \mathcal{L}^{\mathbf{z}}f, \qquad \text{strongly as $n\to\infty$.}
		\end{align*}
		Using \eqref{kernel_cull} together with \eqref{discrete_gen} we can write for each $n\geq 1$
		\begin{align*}
			\overline{\mathcal{L}}^{(\mathbf{z},n)}f(\mathbf{r})=n\left[\E_{(\mathbf{r},\mathbf{z})}\left[f(\mathbf{R}_{\frac{1}{n}\wedge \tau})\right]-f(\mathbf{r})\right], \qquad\text{$\mathbf{r}\in[0,1]^2$.}
		\end{align*}
		On the other hand, by noting that the function $f$ only depends on the first two coordinates $(x_1,x_2)$ of any point $(x_1,x_2,z_1,z_2)\in[0,1]^2\times(0,\infty)^2$ we have by Lemma \ref{mart_prob} 
		\begin{align}\label{mart_gen}
			f(\mathbf{R}_{n^{-1}\wedge \tau})=f(\mathbf{r})+\int_{0}^{n^{-1}\wedge \tau}\sum_{i=1}^2D^i_f(R^1_s,R^2_s,Z^1_s,Z^2_s)ds+M_{n^{-1}\wedge\tau},
		\end{align}
		where for $i=1,2$, $D^i_f$ is defined in Proposition \ref{gen_culling} and $M$ is a local martingale.
		
		Let us define $B(\varepsilon,L):=\{(x_1,x_2)\in\R_+^2: \varepsilon\leq x_i\leq L, \ i=1,2\}$. Then using the fact that $f\in\mathcal{C}^2([0,1]^2)$ and that $(R^1_t,R^2_t,Z^1_t,Z^2_t)\in[0,1]^2\times B(\varepsilon,L)$ for $t\in[0,\tau]$, we can find a constant $C>0$ such that
		\begin{align}\label{bound_con}
			\left|\sum_{i=1}^2D^i_f(R^1_t,R^2_t,Z^1_t,Z^2_t)\right|\leq C\qquad \text{for $t\in[0,\tau]$, $\mathbb{P}$- a.s.}
		\end{align}
		Hence, taking expectations in \eqref{mart_gen} we have
		\begin{align}\label{gen_lim}
			\overline{\mathcal{L}}^{(\mathbf{z},n)}f(\mathbf{r})=n\left[\E_{(\mathbf{r},\mathbf{z})}\left[f(\mathbf{R}_{\frac{1}{n}\wedge \tau})\right]-f(\mathbf{r})\right]=n\E_{(\mathbf{r},\mathbf{z})}\left[\int_{0}^{n^{-1}\wedge \tau}\sum_{1=1}^2D^i_f(R^1_s,R^2_s,Z^1_s,Z^2_s)ds\right].
		\end{align}
		Now, by \eqref{bound_con} 
		\[
		\left|n\int_{0}^{n^{-1}\wedge \tau}\sum_{1=1}^2D^i_f(R^1_s,R^2_s,Z^1_s,Z^2_s)ds\right|\leq C\qquad \text{$\mathbb{P}$-a.s.}
		\]
		Hence, using \eqref{gen_lim} together with dominated convergence gives for $\textbf{r}=(r_1,r_2)\in[0,1]^2$
		\begin{align}\label{point:con}
			\lim_{n\to\infty}\overline{\mathcal{L}}^{(\mathbf{z},n)}f(\mathbf{r})=\lim_{n\to\infty}n\E_{(\mathbf{r},\mathbf{z})}\left[\int_{0}^{n^{-1}\wedge \tau}\sum_{1=1}^2D^i_f(R^1_s,R^2_s,Z^1_s,Z^2_s)ds\right]=\sum_{1=1}^2D^i_f(r_1,r_2,z_1,z_2)=\mathcal{L}^{\mathbf{z}}f(\mathbf{r}).
		\end{align}
		To finish the proof we only need to show that the convergence given in \eqref{point:con} holds in the strong sense, this is achieved by an application of Theorem 1.33 in \cite{Sch}.
	\end{proof}
	
	\section{Moment duality for asymmetric two-island frequency processes.}\label{duality}

	Throughout this section we will prove that the moment dual of the asymmetric two-island frequency process $\mathbf{R}^{(\mathbf{z},\mathbf{r})}$, introduced in Section \ref{ATIFM}, is given by the continuous-time Markov chain $\mathbf{N}^{(\mathbf{z},\mathbf{n})}$ given in Section \ref{duality-intro}. To obtain this result we will rely on martingale techniques, which are quite useful for solving this type of problems. The following proposition is a direct consequence of Theorem 4.11 in Ethier Kurtz \cite{EK} taking $H$ bounded and continuous, and $\alpha=\beta=0$, and can also be seen as a
	small modification of Proposition 1.2 of Jansen and Kurt in \cite{JK}.
	\begin{proposition}\label{prop:provingduality}
		Let $Y^{(1)}=\{Y^{(1)}_t: t\geq 0\}$ and $Y^{(2)}=\{Y^{(2)}_t: t\geq 0\}$ be two Markov processes taking values on $E_1$ and $ E_2$, respectively. Let $H:E_1\times E_2\rightarrow \mathbb{R}$ be a bounded and continuous function and assume that there exist a Borel measurable function $g:E_1\times E_2\mapsto \mathbb{R}$, such that for every $n\in E_1,x\in E_2$ and every $T>0$,  the processes $M^{(1)}=\{M_t^{(1)}: 0\le t\le T\}$ and $M^{(2)}=\{M_t^{(2)}, 0\le t\le T\}$, defined as follows
		\begin{eqnarray}
			M_t^{(1)}:=H(Y^{(1)}_t,x)&-&\int_0^t g(Y^{(1)}_s,x)\ud s\label{MG1}\\
			M_t^{(2)}:=H(n,Y^{(2)}_t)&-&\int_0^t g(n, Y^{(2)}_s)\ud s \label{MG2}
		\end{eqnarray}
		are martingales with respect to  the natural filtration of $Y^{(1)}_t$ and $Y^{(2)}_t$, respectively. Then, 
		the processes $Y^{(1)}$ and $Y^{(2)}$ are dual with respect to $H$. 
	\end{proposition}
	The previous result defines a general duality relation between two Markov processes. Because we are interested in moment duality we will use Proposition \ref{prop:provingduality} with the particular case of $H(x,y,n,m)=x^ny^m$ with $x,y\in[0,1]$ and $n,m\in\mathbb{N}_0$. We assume that $H(x,y,n,m)=0$ for $(x,y)\in[0,1]^2$ and $(n,m)=\Delta$.
	\begin{theorem}\label{dual}
		Assume that $q^{\mathbf{z}}_{(i_1,i_2),(j_1,j_2)}\geq0$ for all $(i_1,i_2),(j_1,j_2)\in\mathbb{N}_0^2\cup\{\Delta\}$. Then for all $\mathbf{r}=(r_1,r_2)\in[0,1]^2$, $\mathbf{n}=(n_1,n_2)\in\mathbb{N}_0^2\cup\{\Delta\}$ and $t>0$
		\begin{equation}
			\E\left[(R_t^{({\mathbf{z}},{\mathbf{r}}),1})^{n_1}(R_t^{({\mathbf{z}},{\mathbf{r}}),2})^{n_2}\right]=\E\left[r_1^{N^{(\mathbf{z},\mathbf{n}),1}_t}r_2^{N^{(\mathbf{z},\mathbf{n}),2}_t}\right].
		\end{equation}
	\end{theorem}
	\begin{proof}
		We will consider $\N_0^2\cup \{\Delta\}$ endowed with the discrete topology and $\N_0^2\cup \{\Delta\}\times [0,1]^2$ with the product topology. We recall that for every fixed $(r_1,r_2)\in[0,1]^2$, $H(r_1,r_2,n_1,n_2)=r_1^{n_1}r_2^{n_2}$ with $H(r_1,r_2,n_1,0)=r_1^{n_1}$, $H(r_1,r_2,0,n_2)=r_2^{n_2}$, $H(r_1,r_2,0,0)=1$  and $H(r_1,r_2,n_1,n_2)=0$ if $(n_1,n_2)=\Delta$, which are bounded and continuous. 
		In addition, for every fixed $\mathbf{n}=(n_1,n_2)\in \N_0^2\cup \{\Delta\}$, $H(r_1,r_2,n_1,n_2)=r_1^{n_1}r_2^{n_2}1_{\{\mathbf{n}\not=\Delta\}}$ is continuous. Therefore, we conclude that $H:\N_0^2\cup \{\Delta\}\times [0,1]^2\mapsto[0,1] $ is continuous.
		
		We observe that $H(\cdot,n_1,n_2)$ is a polynomial in $[0,1]^2$ for fixed $(n_1,n_2)\in\mathbb{N}_0^2\cup\{\Delta\}$. This fact clearly implies that $H(\cdot,n_1,n_2)\in\mathcal{C}^2([0,1]^2)$ and hence it lies in the domain of the generator $\mathcal{L}^{\mathbf{z}}$. Therefore, the process
		\[
		H(R_t^{({\mathbf{z}},{\mathbf{r}}),1},R_t^{({\mathbf{z}},{\mathbf{r}}),2},n_1,n_2)-\int_0^t\mathcal{L}^{\mathbf{z}}H(R_s^{({\mathbf{z}},{\mathbf{r}}),1},R_s^{({\mathbf{z}},{\mathbf{r}}),2},n_1,n_2)ds
		\]
		is a martingale.
		
		Additionally, as in the proof of Lemma 2 in \cite{GPP}, we have that for fixed $(r_1,r_2)\in[0,1]^2$ the function $H(r_1,r_2,\cdot)$ lies in the domain of the generator $\mathcal{Q}^{\mathbf{z}}$, which implies that the process
		\[
		H(r_1,r_2,N^{(\mathbf{z},\mathbf{n}),1}_t,N^{(\mathbf{z},\mathbf{n}),2}_t)-\int_0^t\mathcal{Q}^{\mathbf{z}}H(r_1,r_2,N^{(\mathbf{z},\mathbf{n}),1}_s,N^{(\mathbf{z},\mathbf{n}),2}_s)ds
		\]
		is also a martingale.
		In view of Proposition \ref{prop:provingduality}, we will compute $\mathcal{L}^{\mathbf{z}}H(r_1,r_2,n_1,n_2)$ for $(r_1,r_2)\in[0,1]^2$ and $(n_1,n_2)\in\mathbb{N}_0^2\backslash\{(0,0)\}$. Using Proposition \ref{gen_culling} we obtain that
		\begin{align}\label{dual_0}
			\mathcal{L}^{\mathbf{z}}H(r_1,r_2,n_1,n_2)=\sum_{i=1}^2D^i_H(r_1,r_2,z_1,z_2),
		\end{align}
		with
		\begin{align}\label{dual_1}
			D^1_H(r_1,&r_2,z_1,z_2)\notag\\&=n_1r^{n_1-1}_1r_2^{n_2}(1-r_1)\left(B_{11}^Xr_1+B^X_{12}\frac{z_2}{z_1}r_2+\frac{\beta_1^X}{z_1}\right)-n_1r_1^{n_1}r_2^{n_2}\left(B_{11}^Y(1-r_1)+B^Y_{12}\frac{z_2}{z_1}(1-r_2)+\frac{\beta_1^Y}{z_1}\right)\notag\\
			&+n_2r_1^{n_1}r_2^{n_2-1}(1-r_2)\left(B_{21}^X\frac{z_1}{z_2}r_1+B^X_{22}r_2+\frac{\beta_2^X}{z_2}\right)-n_2r_1^{n_1}r_2^{n_2}\left(B_{21}^Y\frac{z_1}{z_2}(1-r_1)+B^Y_{22}(1-r_2)+\frac{\beta_2^Y}{z_2}\right),\notag\\
			&=(r_1^{n_1+1}r_2^{n_2}-r_1^{n_1}r_2^{n_2})\left[n_1(B_{11}^Y-B_{11}^X)+n_2\frac{z_1}{z_2}(B_{21}^Y-B_{21}^X)\right]+(r_1^{n_1-1}r_2^{n_2}-r_1^{n_1}r_2^{n_2})n_1\frac{\beta^X_1}{z_1}\notag\\
			&+(r_1^{n_1}r_2^{n_2+1}-r_1^{n_1}r_2^{n_2})\left[n_2(B_{22}^Y-B_{22}^X)+n_1\frac{z_2}{z_1}(B_{12}^Y-B_{12}^X)\right]+(r_1^{n_1}r_2^{n_2-1}-r_1^{n_1}r_2^{n_2})n_2\frac{\beta^X_2}{z_2}\notag\\
			&+(r_1^{n_1-1}r_2^{n_2+1}-r_1^{n_1}r_2^{n_2})n_1B^X_{12}\frac{z_2}{z_1}+(r_1^{n_1+1}r_2^{n_2-1}-r_1^{n_1}r_2^{n_2})n_2B^X_{21}\frac{z_1}{z_2}-r_1^{n_1}r_2^{n_2}\left[n_1\frac{\beta_1^Y}{z_1}+n_2\frac{\beta_2^Y}{z_2}\right].
		\end{align}
		and
		\begin{align*}
			D_H^2(r_1,r_2,z_1,z_2)
			&=\frac{c_1^X}{z_1}\left((1-r_1)^2{n_1}(n_1-1)r_1^{n_1-1}r_2^{n_2}-2(1-r_1)n_1r_1^{n_1}r_2^{n_2}\right)\notag\\&+\frac{c_2^X}{z_2}\left((1-r_2)^2n_2(n_2-1)r_1^{n_1}r_2^{n_2-1}-2(1-r_2)n_2r_1^{n_1}r_2^{n_2}\right)\notag\\
			&+\frac{c_1^Y}{z_1}\left((1-r_1)n_1(n_1-1)r_1^{n_1}r_2^{n_2}+2(1-r_1)n_1r_1^{n_1}r_2^{n_2}\right)
			\end{align*}
			\begin{align}\label{dual_2}
			&+\frac{c_2^Y}{z_2}\left((1-r_2)n_2(n_2-1)r_1^{n_1}r_2^{n_2}+2(1-r_2)n_2r_1^{n_1}r_2^{n_2}\right)\notag\\
			&=(r_1^{n_1+1}r_2^{n_2}-r_1^{n_1}r_2^{n_2})\left(\frac{2n_1}{z_1}(c_1^X-c_1^Y)+\frac{n_1(n_1-1)}{z_1}(c_1^X-c_1^Y)\right)\notag\\&+(r_1^{n_1-1}r_2^{n_2}-r_1^{n_1}r_2^{n_2})n_1(n_1-1)\frac{c_1^X}{z_1}\notag\\
			&+(r_1^{n_1}r_2^{n_2+1}-r_1^{n_1}r_2^{n_2})\left(\frac{2n_2}{z_2}(c_2^X-c_2^Y)+\frac{n_2(n_2-1)}{z_2}(c_2^X-c_2^Y)\right)\notag\\&+(r_1^{n_1}r_2^{n_2-1}-r_1^{n_1}r_2^{n_2})n_2(n_2-1)\frac{c_2^X}{z_2}.
		\end{align}
		Therefore using \eqref{dual_1} and \eqref{dual_2} in \eqref{dual_0} we obtain
		\begin{align*}
			\mathcal{L}^{\mathbf{z}}&H(r_1,r_2,n_1,n_2)\notag\\
			&=(r_1^{n_1+1}r_2^{n_2}-r_1^{n_1}r_2^{n_2})\left[n_1(B_{11}^Y-B_{11}^X)+n_2\frac{z_1}{z_2}(B_{21}^Y-B_{21}^X)+\frac{2n_1}{z_1}(c_1^X-c_1^Y)+\frac{n_1(n_1-1)}{z_1}(c_1^X-c_1^Y)\right]\notag\\
			&+(r_1^{n_1}r_2^{n_2+1}-r_1^{n_1}r_2^{n_2})\left[n_2(B_{22}^Y-B_{22}^X)+n_1\frac{z_2}{z_1}(B_{12}^Y-B_{12}^X)+\frac{2n_2}{z_2}(c_2^X-c_2^Y)+\frac{n_2(n_2-1)}{z_2}(c_2^X-c_2^Y)\right]
			\end{align*}
		\begin{align*}
			&+(r_1^{n_1-1}r_2^{n_2}-r_1^{n_1}r_2^{n_2})\left[n_1\frac{\beta^X_1}{z_1}+n_1(n_1-1)\frac{c_1^X}{z_1}\right]+(r_1^{n_1}r_2^{n_2-1}-r_1^{n_1}r_2^{n_2})\left[n_2\frac{\beta^X_2}{z_2}+n_2(n_2-1)\frac{c_2^X}{z_2}\right]\notag\\
			&+(r_1^{n_1-1}r_2^{n_2+1}-r_1^{n_1}r_2^{n_2})n_1B^X_{12}\frac{z_2}{z_1}+(r_1^{n_1+1}r_2^{n_2-1}-r_1^{n_1}r_2^{n_2})n_2B^X_{21}\frac{z_1}{z_2}-r_1^{n_1}r_2^{n_2}\left[n_1\frac{\beta_1^Y}{z_1}+n_2\frac{\beta_2^Y}{z_2}\right]\notag\\&=\mathcal{Q}^{\mathbf{z}}H(r_1,r_2,n_1,n_2).
		\end{align*}
		Meanwhile, for the case $(n_1,n_2)=(0,0)$, we have that $\mathcal{L}^{\mathbf{z}}H(r_1,r_2,0,0)=0=\mathcal{Q}^{\mathbf{z}}H(r_1,r_2,0,0)$ for $(r_1,r_2)\in[0,1]^2$. Similarly, we have that $\mathcal{L}^{\mathbf{z}}H(r_1,r_2,\Delta)=0=\mathcal{Q}^{\mathbf{z}}H(r_1,r_2,\Delta)$ for $(r_1,r_2)\in[0,1]^2$. Therefore, the result follows from Proposition \ref{prop:provingduality}.
	\end{proof}
	

	\section{Large population asymptotics of asymmetric two-island frequency processes.}\label{asymp}
	To disentangle the effects and interactions of the different forms of selection included in our model, we will send the population size to infinity to isolate the deterministic action of selection from the random noise coming from the random genetic drift. To this end, in this section we obtain the deterministic limit obtained when the size of the population in each island grows to infinity. This will allows us to study in the next section the deterministic system that emerges.

	\subsection{Large population limit of asymmetric two-island frequency processes.} 
	Throughout this section, we will assume that the total sizes of the population of each of the islands, given by $z_1$ and $z_2$, are equal to a constant $z>0$. For this reason, we will use $\mathbf{R}^{(z,\mathbf{r})}=(R^{(z,\mathbf{r}),1},R^{(z,\mathbf{r}),2})$ to denote the unique solution to \eqref{sde}, and study the asymptotic behavior of the process $\mathbf{R}^{(z,\mathbf{r})}$ as $z$ increases to infinity. The limiting object, which we denote by $\mathbf{R}^{(\infty,\mathbf{r})}=(R^{(\infty,\mathbf{r}),1},R^{(\infty,\mathbf{r}),2})$, is deterministic and it is the unique solution to the following ordinary differential equation:
	\begin{align}\label{det_lim}
		dR^{(\infty,r),i}_t&=\tilde{b}^i(\mathbf{R}^{(\infty,\mathbf{r})}_t)dt, \qquad t>0, \ i=1,2,
	\end{align}
	with $\mathbf{R}^{(\infty,\mathbf{r})}_0=(r_1,r_2)$ and where 
	\begin{align}\label{fun_b_tilde}
		\tilde{b}^1(\mathbf{x})=(B_{11}^X-B_{11}^Y)x_1(1-x_1)+B_{12}^X(1-x_1)x_2-B_{12}^Yx_1(1-x_2),\notag\\
		\tilde{b}^2(\mathbf{x})=(B_{22}^X-B_{22}^Y)x_2(1-x_2)+B_{21}^X(1-x_2)x_1-B_{21}^Yx_2(1-x_1).
	\end{align}
	for $\mathbf{x}=(x_1,x_2)\in[0,1]^2$.
	
	The next result shows the convergence of the process $\mathbf{R}^{(z,\mathbf{r})}$ to $\mathbf{R}^{(\infty,\mathbf{r})}$ in $p$-th mean, uniformly on compact sets of $\R_+$.
	\begin{theorem}\label{lon}
		For any $T>0$ and $p>1$,
		\begin{align*}
			\lim_{z\to\infty}\E\left[\sup_{t\in[0,T]}\|\mathbf{R}^{(z,\mathbf{r})}_t-\mathbf{R}^{(\infty,\mathbf{r})}_t\|^p\right]=0.
		\end{align*}
	\end{theorem}
	\begin{proof}
		Using \eqref{sde} together with \eqref{det_lim} we obtain for $i=1,2$
		\begin{align}\label{sde_diff}
			R^{(z,\mathbf{r}),i}_t-R^{(\infty,{\mathbf{r}}),i}_t&=\int_0^t\left(\tilde{b}^i(\mathbf{R}^{(z,\mathbf{r})}_s)-\tilde{b}^i(\mathbf{R}^{(\infty,\mathbf{r})}_s)+a^i(\mathbf{R}^{(z,\mathbf{r})}_s)\right)ds\notag\\&+\int_0^t\sqrt{\frac{c_i^X}{z}R^{({z},{\mathbf{r}}),i}_s(1-R^{({z},{\mathbf{r}}),i}_s)^2+\frac{c_i^Y}{z}(R^{({z},{\mathbf{r}}),i}_s)^2(1-R^{({z},{\mathbf{r}}),i}_s)}dW^{i}_s,
		\end{align}
		where $\tilde{b}^i$ is given in \eqref{fun_b_tilde} for $i=1,2$ and
		\begin{align*}
			a^i(\mathbf{x})=\frac{2}{z}(c_i^Y-c_i^X)x_i(1-x_i)+\frac{\beta_i^X}{z}(1-x_i)-\frac{\beta_i^Y}{z}x_i, \qquad \mathbf{x}\in[0,1]^2.
		\end{align*}
		A straightforward computation shows that for $i=1,2$, there exists a constant $C_i>0$ such that for any $\mathbf{x},\mathbf{y}\in[0,1]^2$
		\begin{align*}
			|\tilde{b}^i(\mathbf{x})-\tilde{b}^i(\mathbf{y})+a^{i}(\mathbf{x})|\leq C_i\left(|x_1-y_1|+|x_2-y_2|\right)+\frac{2}{z}|c_i^Y-c_i^X|+\frac{\beta_i^X}{z}+\frac{\beta_i^Y}{z}. 
		\end{align*}
		Hence, using Jensen's inequality and the Burkholder-Davis-Gundy inequality we obtain for some constant $A_i>0$ and $i=1,2$ 
		\begin{align}\label{bound_limit}
			\E\Big[&\sup_{t\in[0,T]}|R^{({z},{\mathbf{r}}),i}_t-R^{(\infty,{\mathbf{r}}),i}_t|^p\Big]\leq 2^p\Bigg[T^{p-1}\E\int_0^T\left|\tilde{b}^i(\mathbf{R}^{(z,\mathbf{r})}_s)-\tilde{b}^i(\mathbf{R}^{(\infty,\mathbf{r})}_s)+a^i(\mathbf{R}^{(z,\mathbf{r})}_s)\right|^pds\notag\\&+A_i\E\left[\int_0^T\left(\frac{c_i^X}{z}R^{({z},{\mathbf{r}}),i}_s(1-R^{({z},{\mathbf{r}}),i}_s)^2+\frac{c_i^Y}{z}(R^{({z},{\mathbf{r}}),i}_s)^2(1-R^{({z},{\mathbf{r}}),i}_s)\right)ds\right]^{p/2}\Bigg]\notag\\
			&\leq K_i\Bigg[T^{p-1}\E\int_0^T\sum_{i=1}^2\left|R^{({z},{\mathbf{r}}),i}_s-R^{(\infty,{\mathbf{r}}),i}_s\right|^pds+\left(\frac{T}{z}\right)^{p}\left(2|c_i^Y-c_i^X|+\beta_i^X+\beta_i^Y\right)^p+\left(\frac{T}{z}\right)^{p/2}(c_i^Y+c_i^X)^{p/2}\Bigg]\notag\\
			&\leq K_i\Bigg[2T^{p-1}\E\int_0^T\sup_{u\in[0,s]}\left\|\mathbf{R}^{(z,\mathbf{r})}_u-\mathbf{R}^{(\infty,\mathbf{r})}_u\right\|^pds+\left(\frac{T}{z}\right)^{p}\left(2|c_i^Y-c_i^X|+\beta_i^X+\beta_i^Y\right)^p+\left(\frac{T}{z}\right)^{p/2}(c_i^Y+c_i^X)^{p/2}\Bigg],
		\end{align}
		for some constant $K_i>0$.
		The previous identity implies that for $i=1,2$, there exists a constant $\tilde{C}_i>0$ such that
		\begin{align*}
			\E\Big[\sup_{t\in[0,T]}&|R^{({z},{\mathbf{r}}),i}_t-R^{(\infty,{\mathbf{r}}),i}_t|^p\Big]\leq 2K_iT^{p-1}\int_0^T\E\left[\sup_{u\in[0,s]}\left\|\mathbf{R}^{(z,\mathbf{r})}_u-\mathbf{R}^{(\infty,\mathbf{r})}_u\right\|^p\right]ds+\tilde{C}_i(1+T^p)\left(\frac{1}{z^{p/2}}+\frac{1}{z^{p}}\right).
		\end{align*}
		Therefore, by an application of Gronwall's lemma we have
		\begin{align*}
			\E\Big[\sup_{t\in[0,T]}&\|\mathbf{R}^{(z,\mathbf{r})}_t-\mathbf{R}^{(\infty,\mathbf{r})}_t\|^p\Big]\leq 2^{p/2}(1+T^p)\sum_{i=1}^2\tilde{C}_i\left(\frac{1}{z^{p/2}}+\frac{1}{z^{p}}\right)e^{2^{p/2+1}T^{p}\sum_{i=1}^2K_i }.
		\end{align*}
		Hence,
		\begin{align}\label{bound_tightness}
			\lim_{z\to\infty}\E\Big[\sup_{t\in[0,T]}&\|\mathbf{R}^{(z,\mathbf{r})}_t-\mathbf{R}^{(\infty,\mathbf{r})}_t\|^p\Big]\leq\lim_{z\to\infty} 2^{p/2}(1+T^p)\sum_{i=1}^2\tilde{C}_i\left(\frac{1}{z^{p/2}}+\frac{1}{z^{p}}\right)e^{2^{p/2+1}T^{p}\sum_{i=1}^2K_i }=0.
		\end{align}
	\end{proof}
	\subsection{Fluctuations of asymmetric two-island frequency processes.} 
	In this section we will study the fluctuations of the process $\mathbf{R}^{(z,\mathbf{r})}$ around its large population limit $\mathbf{R}^{(\infty,\mathbf{r})}$. To this end let us define
	\begin{align}\label{fun_A}
		A^{(1)}_t:&=(B^X_{11}-B^Y_{11})(2R_t^{(\infty,{\mathbf{r}}),1}-1)+B_{12}^XR_t^{(\infty,{\mathbf{r}}),2}-B_{12}^Y(R_t^{(\infty,{\mathbf{r}}),2}-1),\notag\\
		A^{(2)}_t:&=(B^X_{22}-B^Y_{22})(2R_t^{(\infty,{\mathbf{r}}),2}-1)+B_{21}^XR_t^{(\infty,{\mathbf{r}}),1}-B_{21}^Y(R_t^{(\infty,{\mathbf{r}}),1}-1).
	\end{align}
	Let us introduce the process $\mathbf{U}:=\{\mathbf{U}_t=(U^1_t,U^2_t):t\geq0\}$, where for each $i=1,2$ and $t>0$
	\begin{align*}
		U^{i}_t=\int_0^te^{\int_0^sA^{(i)}_udu}\sqrt{c_i^XR^{(\infty,{\mathbf{r}}),i}_s(1-R^{(\infty,{\mathbf{r}}),i}_s)^2+c_i^Y(R^{(\infty,{\mathbf{r}}),i}_s)^2(1-R^{(\infty,{\mathbf{r}}),i}_s)}dW^{i}_s.
	\end{align*}
	We note that the processes $U^1$ and $U^2$ are independent (by the independence between the Brownian motions $W^1$ and $W^2$) and for each $i=1,2,$ $U^i$ is a zero mean Gaussian process with covariance function, $C_{U^i}$, given by 
	\begin{align*}
		C_{U^i}(u,v)=\int_0^{u\wedge v}e^{2\int_0^sA^{(i)}_udu}\left[c_i^XR^{(\infty,{\mathbf{r}}),i}_s(1-R^{(\infty,{\mathbf{r}}),i}_s)^2+c_i^Y(R^{(\infty,{\mathbf{r}}),i}_s)^2(1-R^{(\infty,{\mathbf{r}}),i}_s)\right]ds,\qquad u,v\in\R_+.
	\end{align*}
	Additionally, consider the process $\mathbf{M}^{(z)}:=\{\mathbf{M}^{(z)}=(M^{(z,1)}_t,M^{(z,2)}_t): t\geq0\}$ defined as
	\begin{align*}
		M^{(z,i)}_t:=\int_0^te^{\int_0^sA^{(i)}_udu}\sqrt{c_i^XR^{(z,{\mathbf{r}}),i}_s(1-R^{(z,{\mathbf{r}}),i}_s)^2+c_i^Y(R^{(z,{\mathbf{r}}),i}_s)^2(1-R^{(z,{\mathbf{r}}),i}_s)}dW^{i}_s, \qquad t\geq0, \ i=1,2.
	\end{align*}
	The proof of following auxiliary result is deferred to Appendix \ref{mart_conv_aux}.
	\begin{lemma}\label{mart_conv}
		For any $T>0$,
		\begin{equation}
			\lim_{z\to\infty}\E\left[\sup_{t\in[0,T]}\|\mathbf{M}^{(z)}_t-\mathbf{U}_t\|^2\right]=0.
		\end{equation}
	\end{lemma}
		
	In the next result we will show that the sequence of processes $\{\sqrt{z}(\mathbf{R}^{(z,\mathbf{r})}-\mathbf{R}^{(\infty,\mathbf{r})}):z\geq 1\}$ is tight in the space $\mathbb{C}(\R_+,\R^2)$ i.e. the space of continuous functions from $\R_+$ to $\R^2$ endowed with the topology of uniform convergence in compact sets. 
	\begin{proposition}\label{tightness}
		The family $\{\sqrt{z}(\mathbf{R}^{(z,\mathbf{r})}-\mathbf{R}^{(\infty,\mathbf{r})}):z\geq 1\}$ is tight in the space $\mathbb{C}(\R_+,\R^2)$.
	\end{proposition}
	\begin{proof}
		Fix $T>0$ and let us define $\mathbf{S}^{(z)}:=(S_t^{(z,1)},S_t^{(z,2)})$ with
		\[
		S_t^{(z,i)}:=\sqrt{z}(R^{(z,{\mathbf{r}}),i}_t-R^{(\infty,{\mathbf{r}}),i}_t), \qquad t\geq 0, \ i=1,2.
		\]
		Let $\delta\in(0,1)$ and $\theta\in[0,\delta]$. Let $T>0$ and $(\tau_n)_{n\geq1}$ be a sequence of stopping times such that $0\leq \tau_n<T$.
		Using Jensen's inequality together with the Burkholder-Davis-Gundy inequality and proceeding like in \eqref{bound_limit}, we can find a constant $C_i>0$ such that for $i=1,2$ 
		\begin{align*}
			\E\Big[|S_{\tau_n+\theta}^{(z,i)}-S_{\tau_n}^{(z,i)}&|^4\Big]\leq 2^4z^2\Bigg[\delta^{3}\E\int_{\tau_n}^{\tau_n+\theta}\left|\tilde{b}^i(\mathbf{R}^{(z,\mathbf{r})}_s)-\tilde{b}^i(\mathbf{R}^{(\infty,\mathbf{r})}_s)+a^i(\mathbf{R}^{(z,\mathbf{r})}_s)\right|^4ds\notag\\&+C_i\E\left[\int_{\tau_n}^{\tau_n+\theta}\left(\frac{c_i^X}{z}R^{(z,{\mathbf{r}}),i}_s(1-R^{(z,{\mathbf{r}}),i}_s)^2+\frac{c_i^Y}{z}(R^{(z,{\mathbf{r}}),i}_s)^2(1-R^{(z,{\mathbf{r}}),i}_s)\right)ds\right]^{2}\Bigg]\notag\\
			&\leq K_iz^2\Bigg[\delta^{3}\E\int_{\tau_n}^{\tau_n+\theta}\sum_{i=1}^2\left|R^{(z,{\mathbf{r}}),i}_s-R^{(\infty,{\mathbf{r}}),i}_s\right|^4ds\notag\\&+\left(\frac{\delta}{z}\right)^{4}\left(2|c_i^Y-c_i^X|+\beta_i^X+\beta_i^Y\right)^4+\left(\frac{\delta}{z}\right)^{2}(c_i^Y+c_i^X)^{2}\Bigg]\notag\\
			&\leq K_iz^2\Bigg[2\delta^{4}\E\left[\sup_{u\in[0,T]}\left\|\mathbf{R}^{(z,\mathbf{r})}_u-\mathbf{R}^{(\infty,\mathbf{r})}_u\right\|^4\right]+\left(\frac{\delta}{z}\right)^{4}\left(2|c_i^Y-c_i^X|+\beta_i^X+\beta_i^Y\right)^4\notag\\&+\left(\frac{\delta}{z}\right)^{2}(c_i^Y+c_i^X)^{2}\Bigg],
		\end{align*}
		for some constant $K_i>0$. Now using \eqref{bound_tightness} there exists a constant $C(T)>0$ only depending on $T>0$, such that
		\begin{align}\label{tight_1}
			\E\Big[\|\mathbf{S}^{(z)}_{\tau_n+\theta}-\mathbf{S}^{(z)}_{\tau_n}\|^4\Big]\leq C(T)z^2\delta^{2}\left(\frac{1}{z^{2}}+\frac{1}{z^{4}}\right)\leq 2C(T)\delta^{2}, \qquad \text{for all $z>1$.}
		\end{align}
		Therefore, by Aldous' tightness criterion (see Theorem 16.10 in \cite{Bill}), the estimate \eqref{tight_1} implies that
		the sequence of continuous real processes $\{\sqrt{z}(\mathbf{R}^{(z,\mathbf{r})}-\mathbf{R}^{(\infty,\mathbf{r})}):z\geq 1\}$ is tight in the space $\mathbb{C}(\R_+,\R^2)$.
	\end{proof}
In the next result we show the convergence of the sequence of processes $\{\sqrt{z}(\mathbf{R}^{(z,\mathbf{r})}-\mathbf{R}^{(\infty,\mathbf{r})}):z\geq 1\}$ to a two-dimensional time-inhomogeneous Ornstein Uhlenbeck process. A similar result, when the volatility function is Lipschitz, is provided by Theorem 2.2 in Chapter 2, Section 2 of \cite{W}.
	\begin{theorem}
		For any fixed $T>0$, $\sqrt{z}(\mathbf{R}^{(z,\mathbf{r})}-\mathbf{R}^{(\infty,\mathbf{r})})\to \mathbf{S}^{(\infty)}$ weakly in $\mathbb{C}([0,T],\R)$, where $\mathbf{S}^{(\infty)}:=\{\mathbf{S}^{(\infty)}_t=(S^{(\infty,1)}_t,S^{(\infty,2)}_t):t\geq0\}$ is the unique strong solution to the stochastic differential equation given for $t>0$ by
		\begin{align}\label{sde_tight_3}
			S^{(\infty,1)}_t&=\int_0^t\left[\left(B_{12}^X(1-R_s^{(\infty,\mathbf{r}),1})+B_{12}^YR_s^{(\infty,\mathbf{r}),1}\right)S_s^{(\infty,2)}-A^{(1)}_sS^{(\infty,1)}_s\right]ds\notag\\&+\int_0^t\sqrt{c_1^XR_s^{(\infty,\mathbf{r}),1}(1-R_s^{(\infty,\mathbf{r}),1})^2+c_1^Y(R_s^{(\infty,\mathbf{r}),1})^2(1-R_s^{(\infty,\mathbf{r}),1})}dW^{1}_s,\notag\\
			S^{(\infty,2)}_t&=\int_0^t\left[\left(B_{21}^X(1-R_s^{(\infty,\mathbf{r}),2})+B_{21}^YR_s^{(\infty,\mathbf{r}),2}\right)S_s^{(\infty,1)}-A^{(2)}_sS^{(\infty,2)}\right]ds\notag\\&+\int_0^t\sqrt{c_2^XR_s^{(\infty,\mathbf{r}),2}(1-R_s^{(\infty,\mathbf{r}),2})^2+c_2^Y(R_s^{(\infty,\mathbf{r}),2})^2(1-R_s^{(\infty,\mathbf{r}),2})}dW^{2}_t.
		\end{align}
	\end{theorem}
	\begin{proof}
		From Proposition \ref{tightness}, the family $\{\sqrt{z}(\mathbf{R}^{(z,\mathbf{r})}-\mathbf{R}^{(\infty,\mathbf{r})}):z\geq 1\}$ is relatively compact in $\mathbb{C}(\R_+,\R^2)$. Hence, there exists a subsequence $\{z_n\}_{n\geq 1}$ such that
		$\{\sqrt{z_n}(\mathbf{R}^{(z_n,\mathbf{r})}-\mathbf{R}^{(\infty,\mathbf{r})}):n\geq 1\}$ converges weakly to some process $\mathbf{S}^{(\infty)}:=\{\mathbf{S}^{(\infty)}_t=(S^{(\infty,1)}_t,S^{(\infty,2)}_t):t\geq0\}$ in $\mathbb{C}(\R_+,\R^2)$.
		Therefore, it is enough to prove that there is a unique limit point for any convergent subsequence.
		
		Using integration by parts together with \eqref{sde_diff} we obtain for $i=1,2$ and $t>0$
		\begin{align}\label{conv_1}
			e^{\int_0^tA^{(i)}_sds}S^{(z_n,i)}_t&=\int_0^te^{\int_0^sA^{(i)}_udu}A^{(i)}_sS^{(z_n,i)}_sds+\int_0^te^{\int_0^tA^{(i)}_udu}\sqrt{z_n}\left(\tilde{b}^i(\mathbf{R}^{(z_n,\mathbf{r})}_s)-\tilde{b}^i(\mathbf{R}^{(\infty,\mathbf{r})}_s)+a^i(\mathbf{R}^{(z_n,\mathbf{r})}_s)\right)ds\notag\\
			&+\int_0^te^{\int_0^sA^{(i)}_udu}\sqrt{c_i^XR^{(z_n,\mathbf{r}),i}_s(1-R^{(z_n,\mathbf{r}),i}_s)^2+c_i^Y(R^{(z_n,\mathbf{r}),i}_s)^2(1-R^{(z_n,\mathbf{r}),i}_s)}dW^{i}_s.
		\end{align}
		Using \eqref{fun_A}, straightforward computations give for $s>0$
		\begin{align*}
			&A^{(1)}_sS_s^{(z_n,1)}+\sqrt{z_n}\left(\tilde{b}^1(\mathbf{R}^{(z_n,\mathbf{r})}_s)-\tilde{b}^1(\mathbf{R}^{(\infty,\mathbf{r})}_s)+a^1(\mathbf{R}^{(z_n,\mathbf{r})}_s)\right)\notag\\&=-(B^X_{11}-B^Y_{11})\frac{(S_s^{(z_n,1)})^2}{\sqrt{z_n}}+B_{12}^X\left((1-R_s^{(\infty,\mathbf{r}),1})S_s^{(z_n,2)}-S_s^{(z_n,1)}\frac{S_s^{(z_n,2)}}{\sqrt{z_n}}\right)\notag\\&+B_{12}^Y\left(S_s^{(z_n,1)}\frac{S_s^{(z_n,2)}}{\sqrt{z_n}}+R_s^{(\infty,\mathbf{r}),1}S_s^{(z_n,2)}\right)+\frac{2}{\sqrt{z_n}}(c_1^Y-c_1^X)R_s^{(z_n,\mathbf{r}),1}(1-R_s^{(z_n,\mathbf{r}),1})\notag\\&+\frac{\beta_1^X}{\sqrt{z_n}}(1-R_s^{(z_n,\mathbf{r}),1})-\frac{\beta_1^Y}{\sqrt{z_n}}R_s^{(z_n,\mathbf{r}),1}.
		\end{align*}
		Therefore we can write \eqref{conv_1} for $t>0$ as
		\begin{align}\label{conv_2}
			e^{\int_0^tA^{(1)}_sds}S^{(z_n,i)}_t&-\int_0^te^{\int_0^sA^{(1)}_udu}\Bigg[-(B^X_{11}-B^Y_{11})\frac{(S_s^{(z_n,1)})^2}{\sqrt{z_n}}+B_{12}^X\left(-S_s^{(z_n,1)}\frac{S_s^{(z_n,2)}}{\sqrt{z_n}}+(1-R_s^{(\infty,\mathbf{r}),1})S_s^{(z_n,2)}\right)\notag\\&+B_{12}^Y\left(S_s^{(z_n,1)}\frac{S_s^{(z_n,2)}}{\sqrt{z_n}}+R_s^{(\infty,\mathbf{r}),1}S_s^{(z_n,2)}\right)+\frac{2}{\sqrt{z_n}}(c_1^Y-c_1^X)R_s^{(z_n,\mathbf{r}),1}(1-R_s^{(z_n,\mathbf{r}),1})\notag\\&\hspace{8cm}+\frac{\beta_1^X}{\sqrt{z_n}}(1-R_s^{(z_n,\mathbf{r}),1})-\frac{\beta_1^Y}{\sqrt{z_n}}R_s^{(z_n,\mathbf{r}),1}\Bigg]ds\notag\\&=\int_0^te^{\int_0^sA^{(1)}_udu}\sqrt{c_1^XR_s^{(z_n,\mathbf{r}),1}(1-R_s^{(z_n,\mathbf{r}),1})^2+c_1^Y(R_s^{(z_n,\mathbf{r}),1})^2(1-R_s^{(z_n,\mathbf{r}),1})}dW^{1}_s.
		\end{align}
		Using Skorohod's representation theorem (see  \cite[Theorem 6.7]{Bill}) together with bounded convergence, we obtain that for $t>0$ the left-hand side of \eqref{conv_2} converges weakly as $n\to\infty$ to
		\begin{align*}
			e^{\int_0^tA^{(1)}_sds}S^{(\infty,1)}_t-\int_0^te^{\int_0^sA^{(1)}_udu}\left[B_{12}^X(1-R_s^{(\infty,\mathbf{r}),1})+B_{12}^YR_s^{(\infty,\mathbf{r}),1}\right]S_s^{(\infty,2)}ds.
		\end{align*}
		On the other hand, an application of Lemma \ref{mart_conv} shows that for $t>0$ the right-hand side of \eqref{conv_2} converges weakly as $n\to\infty$ to
		\begin{align*}
			\int_0^te^{\int_0^sA^{(1)}_udu}\sqrt{c_1^XR^{(\infty,\mathbf{r}),1}_s(1-R^{(\infty,\mathbf{r}),1}_s)^2+c_1^Y(R^{(\infty,\mathbf{r}),1}_s)^2(1-R^{(\infty,\mathbf{r}),1}_s)}dW^{1}_s.
		\end{align*}
		Hence, $S^{(\infty,1)}$ satisfies the following
		\begin{align}\label{sde_tight_1}
			e^{\int_0^tA^{(1)}_sds}&S^{(\infty,1)}_t-\int_0^te^{\int_0^sA^{(1)}_udu}\left[B_{12}^X(1-R_s^{(\infty,\mathbf{r}),1})+B_{12}^YR_s^{(\infty,\mathbf{r}),1}\right]S_s^{(\infty,2)}ds\notag\\&=\int_0^te^{\int_0^sA^{(1)}_udu}\sqrt{c_1^XR^{(\infty,\mathbf{r}),1}_s(1-R^{(\infty,\mathbf{r}),1}_s)^2+c_1^Y(R^{(\infty,\mathbf{r}),1}_s)^2(1-R^{(\infty,\mathbf{r}),1}_s)}dW^{1}_s, \qquad t\geq0.
		\end{align}
		Proceeding similarly for the case $i=2$, we obtain
		\begin{align}\label{sde_tight_2}
			e^{\int_0^tA^{(2)}_sds}&S^{(\infty,2)}_t-\int_0^te^{\int_0^sA^{(2)}_udu}\left[B_{21}^X(1-R_s^{(\infty,\mathbf{r}),2})+B_{21}^YR_s^{(\infty,\mathbf{r}),2}\right]S_s^{(\infty,1)}ds\notag\\&=\int_0^te^{\int_0^sA^{(2)}_udu}\sqrt{c_2^XR^{(\infty,\mathbf{r}),2}_s(1-R^{(\infty,\mathbf{r}),2}_s)^2+c_2^Y(R^{(\infty,\mathbf{r}),2}_s)^2(1-R^{(\infty,\mathbf{r}),2}_s)}dW^{2}_s, \qquad t\geq0.
		\end{align}
		Using integration by parts in \eqref{sde_tight_1} and \eqref{sde_tight_2} we obtain that $\mathbf{S}^{(\infty)}$ is solution to \eqref{sde_tight_3}. To conclude the proof, we note that \eqref{sde_tight_3} is a linear stochastic differential equation, then by the classical result on existence and uniqueness of solutions (see Chapter IX, Theorem 2.1 in \cite{RY}) there exists a unique strong solution to \eqref{sde_tight_3}.
		
	\end{proof}
	
	\appendix
	\section{Proof of Lemma \ref{mart_prob}}\label{mart_prob_proof}
		Using the fact that $X^1$ and $X^2$ are semimartingales and that $f$ is twice continuously differentiable on $[0,1]^2\times\R_+^2$ we can apply It\^o's formula (Theorem 3.3 in \cite{RY}) to obtain for $t>0$
	\begin{align}\label{mart}
		f\left(\frac{X^1_{t\wedge\tau}}{X^1_{t\wedge\tau}+Y^1_{t\wedge\tau}},\frac{X^2_{t\wedge\tau}}{X^2_{t\wedge\tau}+Y^2_{t\wedge\tau}},X^1_{t\wedge\tau}+Y^1_{t\wedge\tau},X^2_{t\wedge\tau}+Y^2_{t\wedge\tau}\right)&=f\left(\frac{x^1}{x^1+y^1},\frac{x^2}{x^2+y^2},x^1+y^1,x^2+y^2\right)\notag\\
		&+\int_0^{t\wedge\tau}\sum_{i=1}^3A^i(X^1_s,Y^1_s,X^2_s,Y^2_s)ds+M_{t\wedge\tau}, 
	\end{align}
	where $M=\{M_t:t>0\}$ is a local martingale and
	\begin{align*}
		A^1(x_1,y_1,x_2,y_2):&=\left[((\partial_3f)\circ g)(x_1,y_1,x_2,y_2)+((\partial_1f)\circ g)(x_1,y_1,x_2,y_2)\frac{y_1}{(x_1+y_1)^2}\right]\left(B^X_{11}x_1+B^X_{12}x_2+\beta_1^X\right)\\
		&+\left[((\partial_3f)\circ g)(x_1,y_1,x_2,y_2)-((\partial_1f)\circ g)(x_1,y_1,x_2,y_2)\frac{x_1}{(x_1+y_1)^2}\right]\left(B^Y_{11}y_1+B^Y_{12}y_2+\beta_1^Y\right)\\
		&+\left[((\partial_4f)\circ g)(x_1,y_1,x_2,y_2)+((\partial_2f)\circ g)(x_1,y_1,x_2,y_2)\frac{y_2}{(x_2+y_2)^2}\right]\left(B^X_{21}x_1+B^X_{22}x_2+\beta_2^X\right)\\
		&+\left[((\partial_4f)\circ g)(x_1,y_1,x_2,y_2)-((\partial_2f)\circ g)(x_1,y_1,x_2,y_2)\frac{x_2}{(x_2+y_2)^2}\right]\left(B^Y_{21}y_1+B^Y_{22}y_2+\beta_2^Y\right),
	\end{align*}
	with
	\[
	g(x_1,y_1,x_2,y_2):=\left(\frac{x_1}{x_1+y_1},\frac{x_2}{x_2+y_2},x_1+y_1,x_2+y_2\right).
	\]
	Additionally,
	\begin{align*}
		&A^2(x_1,y_1,x_2,y_2):=c_1^Xx_1\Bigg[((\partial^2_{13}f)\circ g)(x_1,y_1,x_2,y_2)\frac{y_1}{(x_1+y_1)^2}+((\partial^2_{11}f)\circ g)(x_1,y_1,x_2,y_2)\frac{y^2_1}{(x_1+y_1)^4}\\
		&-((\partial_{1}f)\circ g)(x_1,y_1,x_2,y_2)\frac{2y_1}{(x_1+y_1)^3}+((\partial^2_{31}f)\circ g)(x_1,y_1,x_2,y_2)\frac{y_1}{(x_1+y_1)^2}+((\partial^2_{33}f)\circ g)(x_1,y_1,x_2,y_2)\Bigg]\\
		&+c_2^Xx_2\Bigg[((\partial^2_{24}f)\circ g)(x_1,y_1,x_2,y_2)\frac{y_2}{(x_2+y_2)^2}+((\partial^2_{22}f)\circ g)(x_1,y_1,x_2,y_2)\frac{y^2_2}{(x_2+y_2)^4}\\
		&-((\partial_{2}f)\circ g)(x_1,y_1,x_2,y_2)\frac{2y_2}{(x_2+y_2)^3}+((\partial^2_{42}f)\circ g)(x_1,y_1,x_2,y_2)\frac{y_2}{(x_2+y_2)^2}+((\partial^2_{44}f)\circ g)(x_1,y_1,x_2,y_2)\Bigg],
	\end{align*}
	\begin{align*}
		&A^3(x_1,y_1,x_2,y_2):=c_1^Yy_1\Bigg[-((\partial^2_{13}f)\circ g)(x_1,y_1,x_2,y_2)\frac{x_1}{(x_1+y_1)^2}+((\partial^2_{11}f)\circ g)(x_1,y_1,x_2,y_2)\frac{x^2_1}{(x_1+y_1)^4}\\
		&+((\partial_{1}f)\circ g)(x_1,y_1,x_2,y_2)\frac{2x_1}{(x_1+y_1)^3}-((\partial^2_{31}f)\circ g)(x_1,y_1,x_2,y_2)\frac{x_1}{(x_1+y_1)^2}+((\partial^2_{33}f)\circ g)(x_1,y_1,x_2,y_2)\Bigg]\\
		&+c_2^Yy_2\Bigg[-((\partial^2_{24}f)\circ g)(x_1,y_1,x_2,y_2)\frac{x_2}{(x_2+y_2)^2}+((\partial^2_{22}f)\circ g)(x_1,y_1,x_2,y_2)\frac{x^2_2}{(x_2+y_2)^4}\\
		&+((\partial_{2}f)\circ g)(x_1,y_1,x_2,y_2)\frac{2x_2}{(x_2+y_2)^3}-((\partial^2_{42}f)\circ g)(x_1,y_1,x_2,y_2)\frac{x_2}{(x_2+y_2)^2}+((\partial^2_{44}f)\circ g)(x_1,y_1,x_2,y_2)\Bigg].
	\end{align*}
	Now using \eqref{rel_freq} and \eqref{pop_size} we can express \eqref{mart} in terms of the process $(R,Z)=(R^1,R^2,Z^1,Z^2)$ as
	\begin{align}\label{mart_aux_2}
		f(R^1_{t\wedge\tau},R^2_{t\wedge\tau},Z^1_{t\wedge\tau},Z^2_{t\wedge\tau})=f(r^1,r^2,z^1,z^2)+\int_0^{t\wedge\tau}\sum_{i=1}^3B^i(R^1_s,R^2_s,Z^1_s,Z^2_s)ds+M_{t\wedge\tau},\qquad t\geq0,
	\end{align}
	where
	\begin{align}\label{b_1}
		B^1(r_1,r_2,z_1,z_2):&=\partial_1f(r_1,r_2,z_1,z_2)(1-r_1)\left(B^X_{11}r_1+B^X_{12}r_2\frac{z_2}{z_1}+\frac{\beta^X_1}{z_1}\right)\notag\\
		&+\partial_3f(r_1,r_2,z_1,z_2)\left(B^X_{11}z_1r_1+B^X_{12}r_2z_2+\beta^X_1\right)\notag\\
		&-\partial_1f(r_1,r_2,z_1,z_2)r_1\left(B^Y_{11}(1-r_1)+B^Y_{12}(1-r_2)\frac{z_2}{z_1}+\frac{\beta^Y_1}{z_1}\right)\notag\\
		&+\partial_3f(r_1,r_2,z_1,z_2)\left(B^Y_{11}z_1(1-r_1)+B^Y_{12}(1-r_2)z_2+\beta^Y_1\right)\notag\\
		&+\partial_2f(r_1,r_2,z_1,z_2)(1-r_2)\left(B^X_{21}r_1\frac{z_1}{z_2}+B^X_{22}r_2+\frac{\beta^X_2}{z_2}\right)\notag\\
		&+\partial_4f(r_1,r_2,z_1,z_2)\left(B^X_{21}z_1r_1+B^X_{22}r_2z_2+\beta^X_2\right)\notag\\
		&-\partial_2f(r_1,r_2,z_1,z_2)r_2\left(B^Y_{21}(1-r_1)\frac{z_1}{z_2}+B^Y_{22}(1-r_2)+\frac{\beta^Y_2}{z_2}\right)\notag\\
		&+\partial_4f(r_1,r_2,z_1,z_2)\left(B^Y_{21}z_1(1-r_1)+B^Y_{22}(1-r_2)z_2+\beta^Y_2\right),
	\end{align}
	\begin{align}\label{b_2}
		B^2(r_1,r_2,z_1,z_2):&=c_1^Xr_1(1-r_1)\left[\partial_{13}^2f(r_1,r_2,z_1,z_2)+\partial_{31}^2f(r_1,r_2,z_1,z_2)\right]+c_1^Xr_1z_1\partial_{33}^2f(r_1,r_2,z_1,z_2)\notag\\
		&+\frac{c_1^X}{z_1}\left[r_1(1-r_1)^2\partial_{11}^2f(r_1,r_2,z_1,z_2)-2r_1(1-r_1)\partial_{1}f(r_1,r_2,z_1,z_2)\right]\notag\\
		&+c_2^Xr_2(1-r_2)\left[\partial_{24}^2f(r_1,r_2,z_1,z_2)+\partial_{42}^2f(r_1,r_2,z_1,z_2)\right]+c_2^Xr_2z_2\partial_{44}^2f(r_1,r_2,z_1,z_2)\notag\\
		&+\frac{c_2^X}{z_2}\left[r_2(1-r_2)^2\partial_{22}^2f(r_1,r_2,z_1,z_2)-2r_2(1-r_2)\partial_{2}f(r_1,r_2,z_1,z_2)\right],
	\end{align}
	and,
	\begin{align}\label{b_3}
		B^3(r_1,z_1,r_2,z_2):&=-c_1^Yr_1(1-r_1)\left[\partial_{13}^2f(r_1,r_2,z_1,z_2)+\partial_{31}^2f(r_1,r_2,z_1,z_2)\right]+c_1^Y(1-r_1)z_1\partial_{33}^2f(r_1,r_2,z_1,z_2)\notag\\
		&+\frac{c_1^Y}{z_1}\left[r_1^2(1-r_1)\partial_{11}^2f(r_1,r_2,z_1,z_2)+2r_1(1-r_1)\partial_{1}f(r_1,r_2,z_1,z_2)\right]\notag\\
		&-c_2^Yr_2(1-r_2)\left[\partial_{24}^2f(r_1,r_2,z_1,z_2)+\partial_{42}^2f(r_1,r_2,z_1,z_2)\right]+c_2^Y(1-r_2)z_2\partial_{44}^2f(r_1,r_2,z_1,z_2)\notag\\
		&+\frac{c_2^Y}{z_2}\left[r_2^2(1-r_2)\partial_{22}^2f(r_1,r_2,z_1,z_2)+2r_2(1-r_2)\partial_{2}f(r_1,r_2,z_1,z_2)\right].
	\end{align}
	Hence using \eqref{mart_aux_2} together with \eqref{b_1}-\eqref{b_3} we obtain the result.
	\section{Proof of Proposition \ref{exis_uni_sde}}\label{unique}
	\textit{Step 1.-} First we will prove that any solution $\mathbf{R}^{(\mathbf{z},\mathbf{r})}$ to the system \eqref{sde} lies in the set $[0,1]^2$ with probability one. To this end let us denote for $x\in[0,1]$
	\begin{align*}
		\sigma^1(x)&=\sqrt{x(1-x)\left(\frac{c_1^X}{z_1}(1-x)+\frac{c_1^Y}{z_1}x\right)}1_{\{x\in[0,1]\}},\notag\\
		\sigma^2(x)&=\sqrt{x(1-x)\left(\frac{c_2^X}{z_2}(1-x)+\frac{c_2^Y}{z_2}x\right)}1_{\{x\in[0,1]\}}.
	\end{align*}
	Now, we note that the following conditions are satisfied:
	\begin{itemize}
		\item[(i)] For all $x\in\R\backslash[0,1]$ and $i=1,2$ we have that $\sigma^i(x)=0$. 
		\item[(ii)] For $\mathbf{x}=(x_1,x_2)\in\mathbb{R}^2$ with $x_1>1$ we have
		\begin{align*}
			b^1(\mathbf{x})&=-B_{12}^Y\frac{z_2}{z_1}x_1(1-x_2)1_{\{x_2\leq 1\}}+\frac{\beta_1^X}{z_1}(1-x_1)-\frac{\beta_1^Y}{z_1}x_1\leq 0,
		\end{align*}
		additionally for $\mathbf{x}=(x_1,x_2)\in\mathbb{R}^2$ with $x_1<0$ we have
		\begin{align*}
			b^1(\mathbf{x})&=B^X_{12}\frac{z_2}{z_1}(1-x_1)x_21_{\{x_2\geq0 \}}+\frac{\beta_1^X}{z_1}(1-x_1)-\frac{\beta_1^Y}{z_1}x_1\geq 0.
		\end{align*}
		\item[(iii)] For $\mathbf{x}=(x_1,x_2)\in\mathbb{R}^2$ with $x_2>1$ we have
		\begin{align*}
			b^2(\mathbf{x})&=-B^Y_{21}\frac{z_1}{z_2}x_2(1-x_1)1_{\{x_1\leq 1 \}}+\frac{\beta_2^X}{z_2}(1-x_2)-\frac{\beta_2^Y}{z_2}x_2\leq 0.
		\end{align*}
		Additionally for $\mathbf{x}=(x_1,x_2)\in\mathbb{R}^2$ with $x_2<0$ we have
		\begin{align*}
			b^2(\mathbf{x})&=B^X_{21}\frac{z_1}{z_2}x_1(1-x_2)1_{\{x_1\geq 0 \}}+\frac{\beta_2^X}{z_2}(1-x_2)-\frac{\beta_2^Y}{z_2}x_2\geq 0.
		\end{align*}
	\end{itemize}
	Hence, using the previous conditions and proceeding like in the proof of Proposition 2.1 in \cite{FL} we can conclude that $\mathbb{P}(\mathbf{R}^{(\mathbf{z},\mathbf{r})}_t\in[0,1]^2\ \text{ for all $t\geq0$})=1$. 
	
	\textit{Step 2.-} By the results on continuous-type stochastic differential equations (see page 173 in \cite{IW}), there exists a weak solution to \eqref{sde}. Then, in order to prove the existence of a unique strong solution, it remains to prove that the pathwise uniqueness of solutions to \eqref{sde} holds.
	
	To this end, we note that for $\mathbf{x},\mathbf{y}\in[0,1]^2$
	\begin{align*}
		|b^1(\mathbf{x})-b^1(\mathbf{y})|&\leq 2\Bigg[|B_{11}^X-B_{11}^Y|+\frac{z_2}{z_1}(B^X_{12}+B^Y_{12})+\frac{2}{z_1}|c_1^Y-c_1^X|+\frac{\beta_1^X}{z_1}+\frac{\beta_1^Y}{z_1}\Bigg]\left[|x_1-y_1|+|x_2-y_2|\right],
	\end{align*}
	and 
	\begin{align*}
		|b^2(\mathbf{x})-b^2(\mathbf{y})|&\leq 2\Bigg[|B_{22}^X-B_{22}^Y|+\frac{z_1}{z_2}(B^X_{21}+B^Y_{21})+\frac{2}{z_2}|c_2^Y-c_2^X|+\frac{\beta_2^X}{z_2}+\frac{\beta_2^Y}{z_2}\Bigg]\left[|x_1-y_1|+|x_2-y_2|\right].
	\end{align*}
	Additionally for $x,y\in[0,1]$ and $i=1,2,$
	\begin{align*}
		|\sigma^i(x)-\sigma^i(y)|^2&\leq \left|x(1-x)\left(\frac{c_i^X}{z_i}(1-x)+\frac{c_i^Y}{z_i}x\right)-y(1-y)\left(\frac{c_i^X}{z_i}(1-y)+\frac{c_i^Y}{z_i}y\right)\right|\notag\\
		&\leq 3\left(\frac{c_i^X}{z_i}+\frac{c_i^Y}{z_i}\right)|x-y|.
	\end{align*} 
	Therefore by Theorem 2 in \cite{GM} the pathwise uniqueness holds for \eqref{sde}, and hence there exists a unique strong soution to \eqref{sde} (see for instance \cite[p.~104]{SITU}).
	
	\textit{Step 3.-} Finally, we will prove the last assertion in the statement of the proposition. Let us consider $\mathbf{r},\overline{\mathbf{r}}\in[0,1]^2$ and two solutions $\mathbf{R}^{(\mathbf{z},\mathbf{r})}$ and $\mathbf{R}^{(\mathbf{z},\overline{\mathbf{r}})}$ to the stochastic differential equation \eqref{sde} with initial conditions $\mathbf{R}^{(\mathbf{z},\mathbf{r})}_0=\mathbf{r}=(r_1,r_2)$ and $\mathbf{R}^{(\mathbf{z},\overline{\mathbf{r}})}_0=\overline{\mathbf{r}}=(\overline{r}_1,\overline{r}_2)$. Then, proceeding as in the proof of Theorem 2 in \cite{GM} (but considering different initial conditions) we have that there exists a constant $K>0$ such that
	\begin{align*}
		\E\left[\left|R^{(\mathbf{z},\mathbf{r}),i}_t-R^{(\mathbf{z},\overline{\mathbf{r}}),i}_t\right|\right]\leq |r_i-\overline{r}_i|+K\int_0^t\E\left[\left\|\mathbf{R}^{(\mathbf{z},\mathbf{r})}_s-\mathbf{R}^{(\mathbf{z},\overline{\mathbf{r}})}_s\right\|\right]ds,\qquad t\geq0, \ i=1,2.
	\end{align*}
	Hence,
	\begin{align*}
		\E\left[\left\|\mathbf{R}^{(\mathbf{z},\mathbf{r})}_t-\mathbf{R}^{(\mathbf{z},\overline{\mathbf{r}})}_t\right\|\right]&\leq \sum_{i=1}^2|r_i-\overline{r}_i|+2K\int_0^t\E\left[\left\|\mathbf{R}^{(\mathbf{z},\mathbf{r})}_s-\mathbf{R}^{(\mathbf{z},\overline{\mathbf{r}})}_s\right\|\right]ds\notag\\
		&\leq 2|\mathbf{r}-\overline{\mathbf{r}}|+2K\int_0^t\E\left[\left\|\mathbf{R}^{(\mathbf{z},\mathbf{r})}_s-\mathbf{R}^{(\mathbf{z},\overline{\mathbf{r}})}_s\right\|\right]ds,\qquad t\geq0.
	\end{align*}
	Hence, by Gronwall's inequality we obtain \eqref{est_1}.
	\section{Proof of Proposition \ref{gen_culling}}\label{generator}
	(i) Let us denote the semigroup of the process $\mathbf{R}^{(\mathbf{z},\mathbf{r})}$ by $(\mathcal{T}_t)_{t\geq0}$, given for any $f\in\mathcal{C}([0,1]^2)$ by
	\begin{align*}
		\mathcal{T}_tf(\mathbf{r})=\E\left[f(\mathbf{R}^{(\mathbf{z},\mathbf{r})}_t)\right], \qquad \mathbf{r}\in[0,1]^2.
	\end{align*}
	Let us consider $f\in\mathcal{C}^1([0,1]^2)$, then using the multivariate mean value theorem together with \eqref{est_1} 
	\begin{align*}
		|\mathcal{T}_tf(\mathbf{r})-\mathcal{T}_tf(\overline{\mathbf{r}})|=\left|\E\left[f(\mathbf{R}^{(\mathbf{z},\mathbf{r})}_t)\right]-\E\left[f(\mathbf{R}^{(\mathbf{z},\overline{\mathbf{r}})}_t)\right]\right|\leq C(t)\|\nabla f\|_{\infty}\|\mathbf{r}-\overline{\mathbf{r}}\|,
	\end{align*}
	which implies that $\mathbf{r}\mapsto\mathcal{T}_tf(\mathbf{r})$ is continuous. On the other hand, for any $g\in\mathcal{C}([0,1]^2)$ we can find a sequence of functions $(f_n)_{n\geq1}\subset\mathcal{C}^1([0,1]^2)$ such that $f_n\to g$ uniformly on $[0,1]^2$ as $n\to\infty$. Hence,
	\begin{align*}
		\left|\mathcal{T}_tg(\mathbf{r})-\mathcal{T}_tf_n(\mathbf{r})\right|\leq \E\left[\left|g(\mathbf{R}^{(\mathbf{z},\mathbf{r})}_t)-f_n(\mathbf{R}^{(\mathbf{z},\mathbf{r})}_t)\right|\right]\leq \|f_n-g\|_{\infty}, \qquad \mathbf{r}\in[0,1]^2,
	\end{align*}
	therefore $\mathcal{T}_tf_n\to\mathcal{T}_tg$ uniformly on $[0,1]^2$ as well. This implies that $\mathbf{r}\mapsto\mathcal{T}_tg(\mathbf{r})$ is continuous, and hence $\mathcal{T}_t(\mathcal{C}([0,1]^2))\subset \mathcal{C}([0,1]^2)$.
	
	(ii) Lets consider $f\in\mathcal{C}^2([0,1]^2)$, then by an application of It\^o's formula (see Theorem 3.3 in \cite{RY}) together with \eqref{sde} we obtain for $t>0$
	\begin{align}\label{gen_1_proof}
		f(R_t^{(\mathbf{z},\mathbf{r}),1},R_t^{(\mathbf{z},\mathbf{r}),2})=f(r_1,r_2)+\int_0^t\mathcal{L}^{\mathbf{z}}f(R^{(\mathbf{z},\mathbf{r}),1}_s,R^{(\mathbf{z},\mathbf{r}),2}_s)ds+M_t,
	\end{align}
	where $M$ is continuous martingale and $\mathcal{L}^{\mathbf{z}}$ is defined in \eqref{fun_d}.
	
	Using the fact that $f\in\mathcal{C}^2([0,1]^2)$ we can find a constant $C(\mathbf{z})>0$ (dependent only on $\mathbf{z}= (z_1,z_2)$) such that
	\begin{align}\label{bound_gen}
		|\mathcal{L}^{\mathbf{z}}f(r_1,r_2)|\leq C(\mathbf{z}),\qquad r\in[0,1]^2.
	\end{align}
	Hence, taking expectations in \eqref{gen_1_proof} we have
	\begin{align*}
		\sup_{\mathbf{r}\in[0,1]^2}\left|\mathcal{T}_tf(\mathbf{r})-f(\mathbf{r})\right|&=\sup_{\mathbf{r}\in[0,1]^2}\left|\E\left[f(R_t^{(\mathbf{z},\mathbf{r}),1},R_t^{(\mathbf{z},\mathbf{r}),2})\right]-f(r_1,r_2)\right|\notag\\&\leq \sup_{\mathbf{r}\in[0,1]^2}\E\left[\int_0^t\left|\mathcal{L}^{\mathbf{z}}f(R_s^{(\mathbf{z},\mathbf{r}),1},R_s^{(\mathbf{z},\mathbf{r}),2})\right|ds\right]\notag\\&\leq C(\mathbf{z})t\to0,\qquad \text{as $t\to0$.}
	\end{align*}
	This, together with the fact that $\mathcal{T}_t(\mathcal{C}([0,1]^2))\subset \mathcal{C}([0,1]^2)$ implies that $\mathbf{R}^{(\mathbf{z},\mathbf{r})}$ is a Feller process.
	
	(iii) In order to compute the infinitesimal generator of the process $R^{(z,r)}$ we note using \eqref{gen_1_proof}, \eqref{bound_gen}, and dominated convergence
	\begin{align*}
		\lim_{t\downarrow 0}\displaystyle\frac{\E\left[f(R_t^{(\mathbf{z},\mathbf{r}),1},R_t^{(\mathbf{z},\mathbf{r}),2})\right]-f(r_1,r_2)}{t}=\lim_{t\downarrow 0}\E\left[\frac{1}{t}\int_0^t\mathcal{L}^{\mathbf{z}}f(R^{(\mathbf{z},\mathbf{r}),1}_s,R^{(\mathbf{z},\mathbf{r}),2}_s)ds\right]=\mathcal{L}^{\mathbf{z}}f(r_1,r_2).
	\end{align*}
	Hence, using Theorem 1.33 in \cite{Sch} we obtain the result.
		\section{Proof of Lemma \ref{mart_conv}}\label{mart_conv_aux}
		Using Doob's inequality together with It\^o's isometry we obtain for each $i=1,2$
		\begin{align*}
			&\E\left[\sup_{t\in[0,T]}|M^{(z,i)}_t-U^i_t|^2\right]\notag\\&=\E\Bigg[\sup_{t\in[0,T]}\Bigg|\int_0^te^{\int_0^sA^{(i)}_udu}\sqrt{c_i^XR^{(z,{\mathbf{r}}),i}_s(1-R^{(z,{\mathbf{r}}),i}_s)^2+c_i^Y(R^{(z,{\mathbf{r}}),i}_s)^2(1-R^{(z,{\mathbf{r}}),i}_s)}dW^{i}_s\notag\\
			&\hspace{4cm}-\int_0^te^{\int_0^sA^{(i)}_udu}\sqrt{c_i^XR^{(\infty,{\mathbf{r}}),i}_t(1-R^{(\infty,{\mathbf{r}}),i}_s)^2+c_i^Y(R^{(\infty,{\mathbf{r}}),i}_s)^2(1-R^{(\infty,{\mathbf{r}}),i}_s)}dW^{i}_s\Bigg|^2\Bigg]\notag\\
			&\leq 2\E\Bigg[\int_0^Te^{2\int_0^sA^{(i)}_udu}\Bigg[\sqrt{c_i^XR^{(z,{\mathbf{r}}),i}_s(1-R^{(z,{\mathbf{r}}),i}_s)^2+c_i^Y(R^{(z,{\mathbf{r}}),i}_s)^2(1-R^{(z,{\mathbf{r}}),i}_s)}ds\notag\\
			&\hspace{5cm}-\sqrt{c_i^XR^{(\infty,{\mathbf{r}}),i}_t(1-R^{(\infty,{\mathbf{r}}),i}_s)^2+c_i^Y(R^{(\infty,{\mathbf{r}}),i}_s)^2(1-R^{(\infty,{\mathbf{r}}),i}_s)}\Bigg]^2ds\Bigg]\notag\\
			&\leq 6(c_i^X+c_i^Y)\E\int_0^Te^{2\int_0^sA^{(i)}_udu}|R^{(z,{\mathbf{r}}),i}_s-R^{(\infty,{\mathbf{r}}),i}_s|ds\leq 6(c_i^X+c_i^Y)C_i(T)\E\left[\sup_{t\in[0,T]}\|\mathbf{R}^{(z,\mathbf{r})}_t-\mathbf{R}^{(\infty,\mathbf{r})}_t\|\right],
		\end{align*}
		for some constant $C_i(T)>0$ only dependent on $T>0$. Now, by the Cauchy-Schwartz inequality we obtain
		\begin{align*}
			\E\left[\sup_{t\in[0,T]}\|\mathbf{M}^{(z)}_t-\mathbf{U}_t\|^2\right]&\leq 6\sum_{i=1}^2(c_i^X+c_i^Y)C_i(T)\E\left[\sup_{t\in[0,T]}\|\mathbf{R}^{(z,\mathbf{r})}_t-\mathbf{R}^{(\infty,\mathbf{r})}_t\|\right]\notag\\&\leq  6\sum_{i=1}^2(c_i^X+c_i^Y)C_i(T)\E\left[\sup_{t\in[0,T]}\|\mathbf{R}^{(z,\mathbf{r})}_t-\mathbf{R}^{(\infty,\mathbf{r})}_t\|^2\right]^{1/2}.
		\end{align*}
		The result now follows from Theorem \ref{lon}.
	
\end{document}